\documentclass[11pt,a4paper,reqno,oneside]{amsart}
\usepackage[a4paper]{geometry}
\usepackage[english]{babel} 
\usepackage{microtype} 
\usepackage{ragged2e}
\usepackage[autostyle=true]{csquotes}
\usepackage[shortlabels]{enumitem}  
\usepackage{amsthm}
\usepackage{appendix}
\usepackage{mathtools} 
\usepackage{amsfonts}
\usepackage{dsfont}
\usepackage{braket}
\usepackage{xcolor}
\definecolor{royalpurple}{rgb}{0.47, 0.32, 0.66}
\definecolor{dblue}  {RGB}{20,66,129}
\definecolor{ddblue} {RGB}{11,36,69}
\definecolor{jblue}  {RGB}{20,50,100}

\usepackage[unicode=true,
bookmarks=true,bookmarksopen=false,
breaklinks=false,pdfborder={0 0 0},colorlinks=true]
{hyperref}

\usepackage{xcolor}
\definecolor{cblue}{rgb}{0.16, 0.32, 0.75}
\definecolor{cred}{rgb}{0.7, 0.11, 0.11}
\hypersetup{%
	,linkcolor=cred
	,citecolor=cblue
	,urlcolor=black
}

\usepackage[capitalize]{cleveref}

\DeclareMathOperator{\mspan}{Span}
\DeclareMathOperator{\Ker}{Ker}
\DeclareMathOperator{\Ran}{Ran}

\DeclarePairedDelimiterX{\norm}[1]\lVert\rVert{
  \ifblank{#1}{\:\cdot\:}{#1}
}

\newcommand{\iu}{\mathrm{i}\mkern1mu}
\newcommand{\e}{\mathrm{e}}
\newcommand{\nnum}{\mathbb{N}}
\newcommand{\rnum}{\mathbb{R}}
\newcommand{\cnum}{\mathbb{C}}

\newcommand{\znum}{\mathbb{Z}}

\newcommand{\hilbert}{\mathcal{H}}
\newcommand{\domain}{\mathcal{D}}

\numberwithin{equation}{section}

\theoremstyle{plain}
\newtheorem{theorem}{Theorem}[section]
\newtheorem{corollary}[theorem]{Corollary}
\newtheorem{lemma}[theorem]{Lemma}
\newtheorem{proposition}[theorem]{Proposition}
\theoremstyle{definition}
\newtheorem{definition}[theorem]{Definition}
\theoremstyle{remark}
\newtheorem{remark}[theorem]{Remark}
\theoremstyle{remark}
\newtheorem{example}[theorem]{Example}

\AtBeginEnvironment{example}{%
	\pushQED{\qed}%
}
\AtEndEnvironment{example}{\popQED\endexample}
\AtBeginEnvironment{remark}{%
	\pushQED{\qed}%
}
\AtEndEnvironment{remark}{\popQED\endexample}

\AddToHook{env/definition/begin}{\crefalias{theorem}{definition}}
\AddToHook{env/lemma/begin}{\crefalias{theorem}{lemma}}
\AddToHook{env/proposition/begin}{\crefalias{theorem}{proposition}}
\AddToHook{env/remark/begin}{\crefalias{theorem}{remark}}
\AddToHook{env/example/begin}{\crefalias{theorem}{example}}
\AddToHook{env/question/begin}{\crefalias{theorem}{question}}
\AddToHook{env/conjecture/begin}{\crefalias{theorem}{conjecture}}

\usepackage{orcidlink}
\usepackage[foot]{amsaddr}
\newcommand\blfootnote[1]{%
  \begingroup
  \renewcommand\thefootnote{}\footnote{#1}%
  \addtocounter{footnote}{-1}%
  \endgroup
}

\newcommand{\Aop}{A_{k,l}(f)}

\newcommand{\dr}{d^{(n_0)}}
\newcommand{\drpm}{d^{(n_0,\pm)}}
\newcommand{\tcrpm}{\tilde{c}^{(n_0,\pm)}}
\newcommand{\cnpm}{c^{(n_0,\pm)}}

\allowdisplaybreaks

\usepackage[
backend=bibtex,
style=numeric-comp,
giveninits=true,
maxbibnames=299,
natbib=true,
doi=false,
isbn=false,
url=false,
sorting=none,
maxbibnames=8,
date=year
]
{biblatex}
\renewbibmacro{in:}{}
\DeclareFieldFormat[misc]{title}{``#1''\isdot}
\DeclareFieldFormat{journaltitle}{#1\isdot}
\DeclareFieldFormat[article,periodical]{volume}{\mkbibbold{#1}}
\DeclareFieldFormat{pages}{#1}
\AtEveryBibitem{%
	\clearfield{number}}

\title[Self-adjoint realizations of higher-order squeezing operators]{Self-adjoint realizations of higher-order squeezing operators}

\author{Felix Fischer\textsuperscript{\,1}\hspace{2pt}\orcidlink{0009-0001-3040-8184}\hspace{1pt}}

\author{Daniel Burgarth\textsuperscript{\,1}\hspace{2pt}\orcidlink{0000-0003-4063-1264}\hspace{1pt}}

\author{Davide Lonigro\textsuperscript{\,1}\hspace{2pt}\orcidlink{0000-0002-0792-8122}\hspace{1pt}}

\address{\footnotesize \textsuperscript{1}Department Physik, Friedrich-Alexander-Universität Erlangen-Nürnberg, Staudtstraße 7, 91058 Erlangen, Germany}

\email{\footnotesize \href{mailto:felix.o.fischer@fau.de}{\texttt{felix.o.fischer@fau.de}}}
\email{\footnotesize \href{mailto:daniel.burgarth@fau.de}{\texttt{daniel.burgarth@fau.de}}}
\email{\footnotesize \href{mailto:davide.lonigro@fau.de}{\texttt{davide.lonigro@fau.de}}}

\begin{document}

\vspace{-0.75cm}

\begin{abstract}
    Higher-order squeezing captures non-Gaussian features of quantum light by probing moments of the field beyond the variance, and is associated with operators involving nonlinear combinations of creation and annihilation operators. Here we study a class of operators of the form $\xi (a^\dag)^ka^l+\xi^\ast (a^\dag)^la^k+f(a^\dag a)$, which arise naturally in the analysis of higher-order quantum fluctuations. The operators are defined on the linear span of Fock states. We show that the essential self-adjointness of these operators depends on the asymptotics of the real-valued function $f(n)$ at infinity. In particular, pure higher-order squeezing operators ($k\geq3$, $l=0$, and $f(n)=0$) are not essentially self-adjoint, but adding a properly chosen term $f(a^\dag a)$, like a Kerr term, can have a regularizing effect and restore essential self-adjointness. In the non-self-adjoint regime, we compute the deficiency indices and classify all self-adjoint extensions. Our results provide a rigorous operator-theoretic foundation for modeling and interpreting higher-order squeezing in quantum optics, and reveal interesting connections with the Birkhoff--Trjitzinsky theory of asymptotic expansions for recurrence relations.
\end{abstract}

\maketitle\thispagestyle{empty}


\vspace{-0.75cm}


\section{Introduction}\label{sec:intro}

Squeezing of light is a quantum optical phenomenon in which the uncertainty (or quantum noise) of one observable, such as the electric field amplitude at a given phase, is reduced below the level of the standard quantum limit, at the expense of increasing the uncertainty in the conjugate observable. In practical terms, for a squeezed state, the variance of one of the quadratures is smaller than the corresponding one for the vacuum state~\cite{walls-squeezedstateslight-1983,braunstein-generalizedsqueezing-1987}. This enables one to overcome fundamental limitations in measurement processes such as shot noise and quantum radiation pressure noise~\cite{braunstein-generalizedsqueezing-1987,crouch-limitationssqueezingparametric-1988,braunstein-phasehomodynestatistics-1990,braunstein-quantumlimitsprecision-1992,milburn-quantumteleportationsqueezed-1999,braunstein-squeezingirreducibleresource-2005}, and even to enhance the precision of gravitational wave detection~\cite{tse-quantumenhancedadvancedligo-2019,mcculler-frequencydependentsqueezingadvanced-2020,ganapathy-broadbandquantumenhancement-2023}.

At the mathematical level, squeezed states of single-mode (monochromatic) light are generated by applying, either to the vacuum state or to coherent states, the exponential of the squeezing operator
\begin{equation}\label{eq:2-squeeze}
    A_2=\xi(a^\dag)^2+\xi^*a^2,
\end{equation}
where $a,a^\dag$ are the bosonic annihilation and creation operators on a Hilbert space $\hilbert$ satisfying the standard commutation relations $[a,a^\dag]=1$, and $\xi$ is a complex parameter whose phase determines the direction of the squeezing in the phase space. The operator~\eqref{eq:2-squeeze}, defined on the linear span of Fock states---states with a finite number of excitations---is essentially self-adjoint as a consequence of Nelson's analytic vector theorem, cf.~\cite{gorska-squeezingarbitraryorder-2014}, which makes its exponential uniquely defined and unitary. For such states, depending on said phase, either the variance of the position $x$ or the momentum $p$ acquires a value smaller than the one in the vacuum state, at the price of the conjugate variable.

\textit{Higher-order squeezing} refers to generalizations of standard squeezing, where higher-order fluctuations of position and momentum are considered, possibly revealing more subtle forms of nonclassicality that are not captured by conventional squeezing. Definitions of higher-order squeezed states were first proposed by Hong and Mandel~\cite{hong-higherordersqueezingquantum-1985,hong-generationhigherordersqueezing-1985}
and by Hillery and Yu~\cite{hillery-amplificationhigherordersqueezing-1992}, initiating an extensive discussion of the mathematical and physical properties of such states, see e.g. Refs.~\cite{marian-higherordersqueezingphoton-1992,gerry-squeezinghigherordersqueezing-1988,mahran-squeezinghigherordersqueezing-1990,ansari-effectatomiccoherence-1993,marian-higherordersqueezingproperties-1991,peerina-relationsantibunchingsubpoissonian-1984,giri-higherordersqueezingelectromagnetic-2005}.
Among the diverse generalizations discussed in the literature, multiple authors considered $k$-photon generalizations of the squeezing operator in \cref{eq:2-squeeze} in the form
\begin{equation}\label{eq:k-squeeze}
    A_k=\xi(a^\dag)^k+\xi^*a^k,
\end{equation}
with $k\geq3$, together with similar generalizations of well-known models of light--matter interaction such as the Rabi and Jaynes--Cummings model~\cite{fisher-impossibilitynaivelygeneralizing-1984,nagel-higherpowersqueezed-1997,lo-multiquantumjaynescummingsmodel-1998,braunstein-generalizedsqueezing-1987,braunstein-phasehomodynestatistics-1990}. Recent progress in experimentally realizing multi-photon interactions~\cite{corona-thirdorderspontaneousparametric-2011,chang-observationthreephotonspontaneous-2020,menard-emissionphotonmultiplets-2022} reignited the discussion about the theoretical modelling of higher-order photon processes~\cite{lee-exactsolutionsfamily-2011,zhang-solvingtwomodesqueezed-2013,gardas-multiphotonrabimodel-2013,gardas-initialstatesqubit-2013,gorska-squeezingarbitraryorder-2014,zhang-2modekphotonquantum-2017,braak-$k$photonquantumrabi-2025,ashhab-propertiesdynamicsgeneralized-2025}. In the context of quantum information processing and quantum computing, higher-order bosonic operators play a role as an error source~\cite{nagies-rolehigherorderterms-2025,michael-newclassquantum-2016}
and as a resource for state preparation and error mitigation~\cite{cai-bosonicquantumerror-2021,puri-engineeringquantumstates-2017,mirrahimi-dynamicallyprotectedcatqubits-2014}.

This physical motivation leads naturally to the question of whether operators of the form~\eqref{eq:k-squeeze} define legitimate quantum observables—in other words, whether there exists a suitable domain on which the operator $A_k$ is densely defined and essentially self-adjoint.
Unlike the well-understood case $k=2$, where the linear span of Fock states suffices by virtue of Nelson's theorem, the general case does not seem to be settled.
Several authors have argued that the operator $A_k$ fails to be essentially self-adjoint for $k\geq3$, employing a range of methods including the Bargmann space formalism~\cite{zhang-solvingtwomodesqueezed-2013,zhang-2modekphotonquantum-2017,braak-$k$photonquantumrabi-2025}, spectral analysis~\cite{lo-multiquantumjaynescummingsmodel-1998}, and von Neumann’s theory of deficiency indices~\cite{nagel-higherpowersqueezed-1997,gorska-squeezingarbitraryorder-2014}. These arguments have prompted considerable debate and a number of conflicting claims in the literature~\cite{lo-commentsolvingtwomode-2014,lo-commentmultiphotonrabi-2014,gardas-replycommentmultiphoton-2014,lo-commentinitialstates-2014}.

Furthermore, one could consider squeezing operators with additional physical terms, one relevant example being Kerr nonlinearities \cite{mirrahimi-dynamicallyprotectedcatqubits-2014,puri-engineeringquantumstates-2017,cai-bosonicquantumerror-2021,michael-newclassquantum-2016,ayyash-dispersiveregimemultiphoton-2025,ying-criticalquantummetrology-2025,felicetti-universalspectralfeatures-2020} resulting in operators in the form
\begin{equation}\label{eq:k-squeeze-Kerr}
    A_{k,\rm{Kerr}}=\xi(a^\dag)^k+\xi^*a^k+K(a^\dag)^ha^h
\end{equation}
for some integer $h$ and some constant $K\geq0$. Such terms create a natural cutoff for the photon number, and are thus argued to work as regulating terms for numerical simulations \cite{ayyash-dispersiveregimemultiphoton-2025,ying-criticalquantummetrology-2025,ashhab2025truncation}.

\subsection{Our results}

In this paper we aim to shed light on this discussion. To this purpose, we shall investigate self-adjoint realizations of a more general class of operators (cf. \cref{def:higher-order-squeezing}), defined on $\domain_0$, the linear span of Fock states (cf. \cref{eq:d0}), whose action is given by
\begin{equation}\label{eq:most_general}
    \Aop=\xi(a^\dag)^ka^l+\xi^*(a^\dag)^la^k+f(a^\dag a),
\end{equation}
where $k,l$ are arbitrary non-equal positive integers---hereby we fix $k>l$---and $f:\nnum \rightarrow\mathbb{R}_+$ is a nonnegative-valued function. As $a^\dag a$ corresponds to the number operator, $f(a^\dag a)$ is a term only depending on the number of boson excitations.
This additional term accounts for free-field energy contributions or Kerr nonlinearities as in \cref{eq:k-squeeze-Kerr}, which can nontrivially affect the interplay with the squeezing term—much like how the spectrum of the two-photon Rabi model depends sensitively on the ratio between the coupling strength and the field frequency~\cite{duan-twophotonrabimodel-2016}.

In order to illustrate our results, let us first restrict our attention to the case where \(f\) is a polynomial with nonnegative coefficients:
\begin{equation}
    f(n) = a_0 + a_1 n + a_2 n^2 + \dots + a_h n^h, \qquad h \in \mathbb{N}, \quad a_1, \dots, a_{h} \geq 0\, ,
\end{equation}
and $a_h>0$ whenever $f$ is nonzero. The analysis then splits into the following cases:
\begin{itemize}
    \item If \(k + l < 3\), thus necessarily \(l = 0\) and \(k \leq2\), the operator \(\Aop\) is essentially self-adjoint for all choices of the remaining parameters.
    
    \item If \(k + l \geq 3\), the behavior depends on the degree \(h\) of the polynomial \(f\):
    \begin{itemize}
        \item If \(2h > k + l\), or if \(2h = k + l\) and \(\lim_{n \to \infty} f(n) |\xi|^{-1}\,n^{-(k+l)/2} = a_h/|\xi| > 2\), then \(\Aop\) is essentially self-adjoint.
        
        \item If \(2h < k + l\), or if \(2h = k + l\) and \(\lim_{n \to \infty} f(n) |\xi|^{-1}\,n^{-(k+l)/2} = a_h/|\xi|< 2\), then $\Aop$ is not essentially self-adjoint and has equal deficiency indices \(n_\pm = k - l\). In this case, \(\Aop\) admits a \((k - l)\)-dimensional family of self-adjoint extensions, which we can describe explicitly.
    \end{itemize}
\end{itemize}
In particular, all pure higher-order squeezing operators as in \cref{eq:k-squeeze} with $k\geq3$ are not essentially self-adjoint on $\domain_0$ and admit a $k$-dimensional space of self-adjoint extensions. Physically, this means that \cref{eq:k-squeeze} only defines a proper observable once one specifies some additional conditions on the states on which it acts. However, $k$-squeezing operators with an additional term which is a polynomial of order $h$, e.g. the one in \cref{eq:k-squeeze-Kerr}, can be essentially self-adjoint depending on the value of $h$.
Precisely, if $k\geq3$, they are essentially self-adjoint for $2h>k$ and not essentially self-adjoint for $2h<k$, with the critical case $2h=k$ depending on the value of the squeezing parameter $\xi$ and the higher-order coefficient $a_h$ of $f(n)$. This corresponds to the intuition that the higher-order term dominates in the large-photon regime: in particular, if $2h>k$ the field term $f(a^\dag a)$ dominates and has a ``regularizing'' effect. These considerations also have a profound impact on numerical simulations, which is extensively discussed in our companion paper \cite{ashhab2025truncation}.

While above we illustrated our findings for polynomial $f(n)$ for clarity—and because cases of physical interest are expected to be covered by them—our mathematical analysis applies more generally. First: every operator of the form $\Aop$ admits at least one self-adjoint extension. Then, a sufficient condition for essential self-adjointness is  
\begin{equation}
f(n)\,|\xi|^{-1}n^{-(k+l)/2} > \kappa
\end{equation}
for some $\kappa > 2$ and all sufficiently large $n$, in which case the term $f(a^\dag a)$, itself essentially self-adjoint by Nelson’s analytic theorem, dominates over the squeezing term and the result follows perturbatively. Instead, when the squeezing term dominates, von Neumann’s theory reduces the problem to a linear recurrence relation with nonconstant coefficients whose solutions, additionally, must be square summable.  
We analyze this problem using the asymptotic theory of Birkhoff and Trjitzinsky~\cite{birkhoff-formaltheoryirregular-1930,birkhoff-analytictheorysingular-1933} under the assumption that $f(n)\,n^{-(k+l)/2}$ admits an asymptotic expansion in powers of $n^{-1/2}$. To this end, we also provide an explicit adaptation of the Birkhoff--Trjitzinsky theory to asymptotic expansions of this type, which may be of independent interest.

\subsection{Structure of the paper}
The remainder of the paper is organized as follows. In \cref{sec:preliminaries}, we introduce the notation used throughout and define the class of operators under consideration---the $(k,l)$-squeezing operators---along with auxiliary quantities needed in the subsequent analysis. \cref{sec:ess-sa} is devoted to the cases in which these operators are essentially self-adjoint (\cref{prop:k2,prop:rel-bounded}). 
In~\cref{sec:deficiency}, we turn to the more intricate and mathematically richer case where the $(k,l)$-squeezing operators are not essentially self-adjoint. To prepare for this, we show in \cref{sec:conjugation} that all $(k,l)$-squeezing operators have equal deficiency indices and thus admit self-adjoint extensions (\cref{prop:conjugation}). To determine them, we recall key results on linear recurrence relations in~\cref{sec:recurrence}, and then apply this framework in~\cref{sec:recurrence_for_A,sec:deficiency_sub} to compute the deficiency indices and characterize their deficiency subspaces (\cref{thm:a-deficiency}), thereby obtaining a full parametrization of their self-adjoint extensions. These results are finally particularized in \cref{sec:examples} to the case of polynomial $f(n)$ (\cref{prop:polynomial}) and applied to some examples of $k$-order squeezing operators with and without additional polynomial terms. Final considerations and outlooks are gathered in \cref{sec:conclusion}.

\section{Preliminaries}\label{sec:preliminaries}

Throughout the paper, we shall consider a separable Hilbert space $\hilbert$ equipped with an orthogonal basis $(\phi_n)_{n \in \nnum}$. The subspace of $\hilbert$ spanned by finite linear combinations of these vectors,
\begin{equation}\label{eq:d0}
    \domain_0 = \mspan (\phi_n)_{n \in \nnum}\, , 
\end{equation}
is thus dense in $\hilbert$. Physically, we shall interpret each vector $\phi_n$ as the state of a single-mode boson field with exactly $n$ boson excitations---a \textit{Fock state}---so that $\domain_0$ represents the space of all states of the field containing a finite number of excitations.

Correspondingly, we can define the annihilation and creation operators $a,a^\dag$:
\begin{definition}
    \label{def:annihilation-creation}
    The \textit{annihilation operator} $a$ and the \textit{creation operator} $a^\dagger$ are the operators on $\hilbert$ with domains $\domain(a) = \domain(a^\dagger) = \domain_0$ defined as follows: for every $n\in\mathbb{N}$, we set
    \begin{align}\label{eq:action_a}
        a \phi_n & = \sqrt{n}\,  \phi_{n-1} \\\label{eq:action_adag}
        a^\dagger \phi_n & = \sqrt{n+1}\,  \phi_{n+1}.
    \end{align}
\end{definition}
\cref{eq:action_a,eq:action_adag} uniquely determine the action of $a$ and $a^\dag$ on their domain $\domain_0$ by linearity. Notably, $\domain_0$ is invariant under the action of both operators:
\begin{equation}\label{eq:invariance_d0}
    a\domain_0\subset\domain_0,\qquad a^\dag\domain_0\subset\domain_0,
\end{equation}
and thus under the action of arbitrary finite powers of both operators. 
Besides, one readily checks that the operator $a^\dag a:\domain_0\rightarrow\hilbert$ acts as
\begin{equation}
    a^\dag a\phi_n=n\phi_n,
\end{equation}
thus playing the role of the \textit{number operator} of the field.

As mentioned in \cref{sec:intro}, the operators we will investigate also include an additional term in the form $f(a^\dag a)$, thus uniquely depending on the number of bosons. Let us recall its definition:

\begin{definition}
    \label{def:fn}
    Let $f: \nnum \to \rnum_+$ be a nonnegative function, and consider the number operator $a^\dagger a$ on $\domain_0$.
    Then $f(a^\dagger a)$ is the operator with domain $\domain(f(a^\dagger a)) = \domain_0$ acting as
    \begin{equation}
        f(a^\dagger a) \phi_n = f(n) \phi_n 
    \end{equation}
    for every $n\in\mathbb{N}$.
\end{definition}
  Again, one then uniquely recovers the action of $f(a^\dag a)$ on arbitrary vectors of $\domain_0$ by linearity. 

\begin{definition}
    \label{def:higher-order-squeezing}
    Let $k,l \in \nnum$ with $k>l$, $\xi \in \cnum$, and $f: \nnum \to \rnum_+$.
    A $(k,l)$\textit{-order squeezing operator} is an operator $\Aop$ on $\hilbert$ defined by
    \begin{equation}
        \domain(\Aop) = \domain_0\, , \quad \Aop = \xi (a^\dagger)^k a^l +  \xi^\ast (a^\dagger)^l a^k + f(a^\dagger a)\, .
    \end{equation}
\end{definition}
The operator is well-defined by virtue of the invariance of $\domain_0$ under $a,a^\dag$ (cf.  \cref{eq:invariance_d0}); furthermore, $\Aop\domain_0\subset\domain_0$ as well.

For our purposes, we will need to exactly determine the action of such operators. To this end, let us introduce the following coefficients:
\begin{definition}
    \label{def:beta}
    Given $n,k,l \in \nnum$, we define
    \begin{equation}
        \beta^{kl}_n = \sqrt{(n-l+1,l)(n-l+1,k)} \, , 
    \end{equation}
    where, for every $x,s\in\mathbb{N}$, $(x,s)$ is the \textit{Pochhammer symbol} given by
    \begin{equation}\label{eq:pochhammer}
        (x,s) = \prod_{i=1}^s(x+i-1) = x (x+1) \dots (x+s-1),
    \end{equation}
    and we use the convention $(x,s) = 0$ whenever $x$ is a negative integer.
\end{definition}
These coefficients fulfill the following properties:
\begin{lemma}
    \label{lem:beta-sym}
    Let $n,k,l \in \nnum$.
    Then    
    \begin{equation}\label{eq:symmetry}
        \beta^{kl}_{n-k+l} = \beta^{lk}_n \quad \mathrm{and} \quad  \beta^{lk}_{n-l+k} = \beta^{kl}_n \, .
    \end{equation}
    Furthermore, $\beta^{kl}_n =0$ for all $n < l$, and $\beta^{lk}_n = 0$ for all $n < k$.
\end{lemma}
\begin{proof}
   \cref{eq:symmetry} follows from a direct computation:
    \begin{align}
        \beta^{kl}_{n-k+l} & = \sqrt{(n-k+l-l+1,l)(n-k+l-l+1,k)} = \beta^{lk}_n \, , \\
        \beta^{lk}_{n-l+k} & = \sqrt{(n-l+k-k+1,k)(n-l+k-k+1,l)} = \beta^{kl}_n \, .
    \end{align}
    If $l<n$, $n-l+1\leq 0$ and therefore $(n-l+1,l) = 0$.
    Thus, $\beta^{kl}_n = 0$ for $l < n$.
    Exchanging $k$ and $l$, we obtain $\beta^{lk}_n = 0$ for $n < k$. 
\end{proof}

\begin{lemma}
    \label{lem:a-action}
    Let $k,l\in\mathbb{N}$ with $k>l$, $\xi\in\mathbb{C}$, $f:\nnum\rightarrow\mathbb{R}_+$, and $\Aop$ the corresponding $(k,l)$-squeezing operator. Then, for all $n\in\mathbb{N}$,
    \begin{equation}\label{eq:action}
        \Aop \phi_n = \xi \beta^{kl}_n  \phi_{n+(k-l)} +f(n)\phi_n +  \xi^\ast \beta^{lk}_n  \phi_{n-(k-l)}\, .
    \end{equation}
\end{lemma}
\begin{remark}\label{rem:abuse}
    Note that \cref{eq:action} involves a slight abuse of notation: it contains $\phi_n$ with possibly negative values of the index $n$ up to $-(k-l)$, which are not defined. However, such terms are always zero since the corresponding multiplicative coefficients $\beta^{lk}_{n}$ vanish by \cref{lem:beta-sym}. This abuse of notation enables us to write the action of $A_{k,l}(f)$ on Fock states in a unified fashion without explicitly distinguishing between the cases $n<k-l$ and $n\geq k-l$. Equivalently, one could adopt the convention $\phi_n=0$ for $n<0$.
\end{remark}
\begin{proof}
    Let $n,r\in \nnum$.
    Applying~\cref{def:annihilation-creation} $r$ times, we obtain
    \begin{align}
        a^r \phi_n & = 0,&&    n<r \, ; \\
        a^r \phi_n & = \sqrt{n (n-1) \dots (n-r+1)}\, \phi_{n-r}  = \sqrt{(n-r+1,r)} \, \phi_{n-r}, &&  n \geq r\, ; \\
        (a^\dagger)^r \phi_n & = \sqrt{(n+1) (n+2) \dots  (n+r)}\, \phi_{n+r}  = \sqrt{(n+1,r)}\,  \phi_{n+r}, &&  n \in \nnum\, .
    \end{align}
    Thus, using~\cref{def:beta}, we get
    \begin{align}
        \xi (a^\dagger)^k a^l \phi_n & = 0, && n<l \, ; \\
        \xi (a^\dagger)^k a^l \phi_n & = \xi \sqrt{(n-l+1,l)(n-l+1,k)} \phi_{n-l+k} = \xi \beta^{kl}_n \phi_{n-l+k}, && n \geq l \, .
    \end{align}
    As $\beta^{kl}_n = 0$ for $n < l$ by \cref{lem:beta-sym}, the previous statement can thus be summarized as
    \begin{equation}
        \xi (a^\dagger)^k a^l \phi_n = \xi \beta_n^{kl} \phi_{n-l+k} \, .
    \end{equation}
Similarly, $\beta^{lk}_n = 0$ for $n < k$, whence, recalling \cref{rem:abuse}, we get
\begin{equation}
        \xi^\ast (a^\dagger)^l a^k \phi_n = \xi^\ast \sqrt{(n-k+1,k)(n-k+1,l)} \phi_{n-k+l} = \xi^\ast \beta^{lk}_n \phi_{n-k+l}\, .
    \end{equation}
    The claim then follows from these expressions and $f(a^\dag a)\phi_n=f(n)\phi_n$.
\end{proof}
\begin{proposition}
    \label{prop:a-symmetric}
    For every $k>l\in\mathbb{N}$, $\xi\in\mathbb{C}$, and $f:\nnum\rightarrow\mathbb{R}_+$, the $(k,l)$-squeezing operator $\Aop$ is symmetric.
\end{proposition}
\begin{proof}
    Let $n,m \in \nnum$.
    Using~\cref{lem:a-action,lem:beta-sym}, we obtain
    \begin{align}
        \braket{\phi_m, \Aop\phi_n} & = \braket{\phi_m,\xi \beta^{kl}_n \phi_{n-l+k}+f(n) \phi_n +\xi^\ast \beta^{lk}_n \phi_{n-k+l}} \\
        & = \xi \beta^{kl}_n \delta_{m,n-l+k}+f(n)\delta_{mn} + \xi^\ast \beta^{lk}_n \delta_{m,n-k+l} \\
        & = \xi \beta^{kl}_{m-k+l} \delta_{m-k+l,n}+f(n)\delta_{mn} + \xi^\ast \beta^{lk}_{m-l+k} \delta_{m-l+k,n} \\
        & = \braket{\xi^\ast \beta^{lk}_{m} \phi_{m-k+l}+f(n)\phi_m + \xi \beta^{kl}_{m} \phi_{m-l+k},\phi_n } = \braket{\Aop \phi_m,\phi_n}\,.
    \end{align}
    By linearity, $\braket{\varphi,\Aop\psi} = \braket{\Aop \varphi,\psi}$ for all $\varphi,\psi \in \domain_0$.
\end{proof}

As anticipated, our main goal will be to determine, for different parameter regimes, whether $\Aop$ is essentially self-adjoint---and, if not, whether it admits self-adjoint extensions. In preparation for this, let us determine bounds for $\beta^{kl}_n$:

\begin{lemma}
    \label{lem:beta-bounds}
    Let $k > l \in \nnum$, $n\in\mathbb{N}$ and $\beta^{kl}_n$ as per~\cref{def:beta}.
    Then, for every pair of positive numbers $C_1<1$ and $C_2 > 1$, there exists $N \in \nnum$  such that, for all $n \geq N$,
    \begin{equation}
        C_1 n^{(k+l)/2} \leq \beta^{lk}_n < \beta^{kl}_n \leq C_2 n^{(k+l)/2}.
    \end{equation}
\end{lemma}
\begin{proof}
    For every $x<y \in \nnum$ and $s \in \nnum$, the following relation between the Pochhammer symbols $(x,s)$ and $(y,s)$ holds:
    \begin{equation}
        (x,s) = \prod_{i = 1}^s (x+i-1) < \prod_{i = 1}^s (y+i-1) = (y,s) .
    \end{equation}
   Consequently, as $n-k+1<n-l+1$ by assumption, 
    \begin{equation}
        \beta^{lk}_n = \sqrt{(n-k+1,k)(n-k+1,l)} < \sqrt{(n-l+1,l)(n-l+1,k)} = \beta^{kl}_n
    \end{equation}
    holds for all $n \geq k$.    
    Furthermore, for $x,s \in \nnum$, the following upper and lower bounds on $(x,s)$ hold:
    \begin{equation}
        \label{proofeq:lem-beta-bounds-pochhammer}
        x^s \leq (x,s) = \prod_{i =1}^s (x+i-1) \leq (x+s-1)^s \, .
    \end{equation}
    Thus, we obtain
    \begin{equation}
        \beta^{lk}_n = \sqrt{(n-k+1,k)(n-k+1,l)} \geq (n-k+1)^{(k+l)/2} = n^{(k+l)/2} \left(\frac{n-k+1}{n}\right)^{(k+l)/2} \, .
    \end{equation}
    But $\lim_{n \to \infty} (n-k+1)/n = 1$, hence for all $C_1 < 1$ there exists a $N_1 \in \nnum$ such that, for all $n \geq N_1$,
    \begin{equation}
        \left(\frac{n-k+1}{n}\right)^{(k+l)/2} \geq C_1 \, , 
    \end{equation}
    and therefore
    \begin{equation}
        \beta^{lk}_n \geq n^{(k+l)/2} \left(\frac{n-k+1}{n}\right)^{(k+l)/2} \geq C_1  n^{(k+l)/2} \, .
    \end{equation}
    Similarly, $\lim_{n \to \infty} (n+k)/n = 1$, and thus for all $C_2 > 1$ there exists $N_2 \in \nnum$ such that for all $n \geq N_2$
    \begin{equation}
        \left(\frac{n+k}{n}\right)^{(k+l)/2} \leq C_2 \, , 
    \end{equation}
    and using~\cref{proofeq:lem-beta-bounds-pochhammer} we obtain
    \begin{equation}
        \beta^{kl}_n = \sqrt{(n-l+1,l)(n-l+1,k)} \leq (n+k)^{(k+l)/2} = n^{(k+l)/2}  \left(\frac{n+k}{n}\right)^{(k+l)/2} \leq C_2 n^{(k+l)/2} \, . 
    \end{equation}
    The claimed estimate thus holds for $N = \max(N_1,N_2)$.
\end{proof}

\section{Essentially self-adjoint squeezing operators}\label{sec:ess-sa}

Let us begin by considering the cases $k+l\leq3$, i.e. $l=0$ and $k=1,2$. In the absence of the additional free field term $f(a^\dag a)$, one obtains, respectively for $k=1$ and $k=2$, the displacement operator
\begin{equation}
    A_{1,0}=\xi a^\dag+\xi^* a,
\end{equation}
which is essentially self-adjoint, and the standard squeezing operator
\begin{equation}
    A_{2,0}=\xi (a^\dag)^2+\xi^* a^2,
\end{equation}
which is also known to be essentially self-adjoint on $\domain_0$; this can be proven by using Nelson's analytic vector theorem, 
see e.g. Ref.~\cite{gorska-squeezingarbitraryorder-2014}.
For a review on Nelson's analytic vector theorem and generalizations, see e.g.~\cite[Theorem X.39]{reed-mmmp2-fourier-1975}.

As a first result, we show that adding a field term $f(a^\dagger a)$ which grows in the same order as $(a^\dagger)^2$ does not prevent $\Aop$ from being essentially self-adjoint:
\begin{proposition}
    \label{prop:k2}
    Let $ k \leq 2, \,\xi \in \cnum$, and $f: \nnum \to \rnum_+$.
    Suppose that there exists $C > 0$ such that $f(n) \leq C \beta^{k0}_n$ for all $n \in \nnum$.
    Then the $(k,0)$-squeezing operator $A_{k,0}(f)$ is essentially self-adjoint.
\end{proposition}
\begin{proof}
    To simplify the notation, we rescale $C$ in such a way that $f(n) \leq C |\xi| \beta^{k0}_n$ for all $n \in \nnum$.
    As $A_{k,0}(f) \domain_0 \subset \domain_0$, the Fock states are in the domain of any power of $A_{k,0}(f)$, i.e. $\phi_n \in \domain(\Aop^s) = \domain_0$ for all $s \in \nnum$. In the following we show that, in fact, each Fock state is an analytic vector for $A_{k,0}(f)$: there exists $t>0$ such that 
    \begin{equation}
        \sum_{s = 0}^{\infty} \frac{\norm{A_{k,0}(f)^s \phi_n }t^s}{s!} < \infty \, .
    \end{equation}
    We begin by noticing that $A_{k,0}(f)^s \phi_n$ can be written as the sum of $3^s$ terms of the kind
    \begin{equation}
        \label{proofeq:analytic-bi-prod}
        \left( \prod_{i = 1}^s b_i \right) \phi_n   \, , \qquad b_i \in \{\xi^\ast a^k,\xi (a^\dagger)^k,f(a^\dagger a)\} \, .
    \end{equation}
    For each $i \in [1,s]$ and $m \in \nnum$, we have 
    \begin{equation}
    \label{proofeq:analytic-terms}
        b_i \phi_m \in \left\{
            \xi  \beta^{k0}_m \phi_{m+k} \, , 
            f(m) \phi_m \,, 
            \xi^\ast \beta^{0k}_m \phi_{m-k} \right\} \, , 
    \end{equation}
    i.e. each $b_i$ shifts a basis vector $\phi_m$ by at most $k \in \{1,2\}$ up or down.
    Furthermore, by assumption $f(m) \leq C |\xi| \beta^{k0}_m$  and, per~\cref{def:beta}, the $\beta^{k0}_m$ are monotonously increasing in $m$.
    Therefore, the biggest multiplicative factor appearing in~\cref{proofeq:analytic-terms} is upper bounded by $C |\xi| \beta^{k0}_{n+sk}$. Consequently, since $k\leq2$,
    \begin{equation}
        \beta^{k0}_{n+sk} \leq \beta^{k0}_{n+2s} = \sqrt{(n+2s+1,k)} \leq \sqrt{(n+2s+1,2)} \leq (n+2s+2)^{2/2} \, . 
    \end{equation}
    Hence, each of the $3^s$ terms from~\cref{proofeq:analytic-bi-prod} is bounded by 
    \begin{equation}
        \norm*{\prod_{i = 1}^s b_i \phi_n} \leq  C^s |\xi|^s (n+2s+2)^s \, ,
    \end{equation}
    and applying the triangle inequality we obtain
    \begin{align}
        \norm{A_{k,0}(f)^s \phi_n} &  \leq (3C |\xi|)^s (n+2s+2)^s = (3C |\xi|)^s \,  \sum_{p=0}^{s}\binom{s}{p}(n+2)^p (2s)^{s-p} \\
        &  \leq (3C |\xi|)^s 2^s (n+2)^s (2s)^s = (12 (n+2) C |\xi| )^s s^s \, . 
    \end{align}
    Choosing $t = (12 C (n+2) |\xi|)^{-1} \e^{-2}$, we thus get
    \begin{equation}
        \sum_{s = 0}^{\infty} \frac{\norm{A_{k,0}(f)^s \phi_n }t^s}{s!} \leq \sum_{s = 0}^{\infty}\frac{\e^{-2s}s^s}{s!} \, . 
    \end{equation}
    To determine whether the series at the right-hand side converges, define $c_s = \e^{-2s} s^s/s!$.
    Then
    \begin{equation}
        \lim_{s \to \infty} \frac{c_{s+1}}{c_s}  = \e^{-2} \lim_{s \to \infty}\frac{(s+1)^{(s+1)} s!}{s^s (s+1)!} = \e^{-2} \lim_{s \to \infty}\frac{(s+1)^s }{s^s} = \e^{-2} \lim_{s \to \infty}\left(1+\frac{1}{s}\right)^s = \e^{-1} < 1\, , 
    \end{equation}
    hence the series $\sum_{s=0}^{\infty}c_s$ converges by the ratio test.

We proved that every Fock state $\phi_n$ is an analytic vector of $A_{k,0}(f)$. As such, every element of the dense space $\domain_0$ is also analytic. By Nelson's analytic vector theorem, $A_{k,0}(f)$ is essentially self-adjoint.
\end{proof}

A simplified version of~\cref{prop:k2} with more explicit assumptions on $f(n)$ can be readily obtained by taking into account the bounds for $\beta^{k0}_n$ from~\cref{lem:beta-bounds}:
\begin{corollary}\label{cor:slower}
    Let $k \leq 2,\, \xi \in \cnum$, and $f: \nnum \to \rnum_+$.
    Suppose that there exists $C > 0$ and $N \in \nnum$ such that $f(n) \leq C n^{k/2}$ for all $n \geq N$.
    Then $A_{k,0}(f)$ is essentially self-adjoint.
\end{corollary}
\begin{proof}
    By \cref{def:beta}, $n^{k/2} \leq \beta^{k0}_n$ for all $n\in \nnum$.
    Setting $C_1  = \max_{n \leq N} f(n)/\beta^{k0}_n$ and $C_2 = \max(C,C_1)$, we thus have $f(n) \leq C_2 \beta^{k0}_n$ for all $n \leq N$, and 
    \begin{equation}
        f(n) \leq C n^{k/2} \leq C \beta^{k0}_n \leq C_2 \beta^{k0}_n \quad \forall n \geq N,
    \end{equation}
    whence the claim follows by~\cref{prop:k2}.
\end{proof}

We showed that, for $l=0$ and $k\leq2$, $\Aop$ is essentially self-adjoint provided that $f(a^\dagger a)$ grows more slowly than the squeezing part of $\Aop$. We will now show that, for arbitrary $k$ and $l$, $\Aop$ is also essentially self-adjoint provided that $f(a^\dagger a)$ grows \textit{faster} with $n$ than the squeezing part. The remaining cases where $f(a^\dagger a)$ grows more slowly but $k+l\geq3$ will be analyzed in \cref{sec:deficiency}.

To this end, we first show that the operator $f(a^\dagger a)$ is essentially self-adjoint on $\domain_0$ for an arbitrary function $f(n)$.
\begin{lemma}
    \label{lem:f-esssa}
    Let $f:\nnum \to \rnum_+$.
    Then the operator $f(a^\dagger a)$ as per~\cref{def:fn} is essentially self-adjoint.
\end{lemma}
\begin{proof}
    Clearly $f(a^\dagger a)$ is symmetric and, by construction, all $\phi_n$ are eigenvectors of $f(a^\dagger a)$.
    Thus, $f(a^\dagger a)$ has an orthogonal basis of eigenvectors, and by~\cite[Theorem 2.21]{teschl-mathematicalmethodsquantum-2009} $f(a^\dagger a)$ is essentially self-adjoint.
 \end{proof}
 Next, we show that the operator $\xi (a^\dagger)^k a^l$ is, for sufficiently fast growing $f(n)$, relatively bounded with respect to $f(a^\dagger a)$.
\begin{lemma}
    \label{lem:rel-bounded}
    Let $k > l\in\nnum$, $\xi \in \cnum$, and $f: \nnum \to \rnum_+$. 
    Assume that there exist $\kappa > 0$ and $N \in \nnum$ such that, for all $n \geq N$, $f(n) \geq \kappa |\xi| \beta^{kl}_n$.
    Then the following estimate holds for all $\psi \in \domain_0$:
    \begin{equation}
        \norm{\xi (a^\dagger)^k a^l \psi+\xi^\ast (a^\dagger)^l a^k \psi} \leq \frac{2}{\kappa} \norm{f(a^\dagger a)\psi}+ 2|\xi| \beta^{kl}_{N-1} \norm{\psi}\, .
    \end{equation}
\end{lemma}
\begin{proof}
    Define $b = 2|\xi| \beta^{kl}_{N-1}$ and let $\psi \in \domain_0$.
    Thus, there exists $M \in \nnum$ such that $\psi = \sum_{m=0}^{M} c_m \phi_m $.
    By setting additional coefficients to zero, we can assume $M \geq N$ without loss of generality.
    By~\cref{lem:a-action}, we obtain
    \begin{align}
        \label{proofeq:rel-bound-1}
        \norm{(a^\dagger)^k a^l \psi}^2 & = \sum_{n=l}^{M} |c_n|^2 (\beta^{kl}_n)^2 \, , \\
        \norm{(a^\dagger)^l a^k \psi}^2 & = \sum_{n=k}^{M} |c_n|^2 (\beta^{lk}_n)^2 \, , \\        
        \norm{f(a^\dagger a)\psi}^2 & = \sum_{n=0}^{M} |c_n|^2 f(n)^2 \, .
    \end{align}
    Combining~\cref{lem:beta-bounds} and our growth assumption, we get
    \begin{equation}
        |\xi| \beta^{lk}_n \leq |\xi| \beta^{kl}_n \leq \frac{1}{\kappa} f(n) \quad \forall n \geq N\, .
    \end{equation}
    Furthermore, both $\beta^{kl}_n$ and $\beta^{lk}_n$ are monotonously increasing in $n$, whence
    \begin{align}
        \norm{\xi (a^\dagger)^k a^l \psi}^2 & \leq \sum_{n=l}^{N-1} |c_n|^2 |\xi|^2 (\beta^{kl}_{N-1})^2 + \sum_{n=N}^{M} |c_n|^2 \frac{1}{\kappa^2} f(n)^2 \leq \frac{b^2}{4} \norm{\psi}^2 + \frac{1}{\kappa^2} \norm{f(a^\dagger a) \psi}^2, \\
        \norm{\xi (a^\dagger)^l a^k \psi}^2 & \leq \sum_{n=k}^{N-1} |c_n|^2 |\xi|^2 (\beta^{lk}_{N-1})^2 + \sum_{n=N}^{M} |c_n|^2 \frac{1}{\kappa^2} f(n)^2 \leq \frac{b^2}{4} \norm{\psi}^2 + \frac{1}{\kappa^2} \norm{f(a^\dagger a) \psi}^2,
    \end{align}
    and using the sublinearity of the square root we get
    \begin{align}
        \norm{\xi (a^\dagger)^k a^l \psi} & \leq  \left(\frac{b^2}{4} \norm{\psi}^2 + \frac{1}{\kappa^2} \norm{f(a^\dagger a) \psi}^2\right)^{1/2} \leq \frac{b}{2} \norm{\psi} + \frac{1}{\kappa} \norm{f(a^\dagger a)  \psi} \, , \\
        \norm{\xi^\ast (a^\dagger)^l a^k \psi} & \leq  \left(\frac{b^2}{4} \norm{\psi}^2 + \frac{1}{\kappa^2} \norm{f(a^\dagger a) \psi}^2\right)^{1/2} \leq \frac{b}{2} \norm{\psi} + \frac{1}{\kappa} \norm{f(a^\dagger a) \psi} \, .
    \end{align}
    Applying the triangle inequality, we obtain the claimed result.
\end{proof}
We aim to use~\cref{lem:rel-bounded} to prove that $\Aop$ is essentially self-adjoint as a consequence of the self-adjointness of $f(a^\dag a)$. If $\kappa<2$, then the squeezing term is indeed relatively bounded with relative bound smaller than one, thus the conditions of the Kato--Rellich theorem (see e.g.~\cite[Theorem X.12]{reed-mmmp2-fourier-1975}) are fulfilled and $\Aop$ is essentially self-adjoint. However, we can include the case $\kappa=2$ by using a theorem by Wüst (see e.g.~\cite[Theorem 5.30]{weidmann-linearoperatorshilbert-2013}):
\begin{proposition}
    \label{prop:rel-bounded}
    Let $k > l\in\nnum$, $\xi \in \cnum$, and $f: \nnum \to \rnum_+$. 
    Assume that there exists $\kappa \geq 2$ and $N \in \nnum$ such that, for all $n \geq N$, $f(n) \geq \kappa |\xi| \beta^{kl}_n$.
    Then the $(k,l)$-squeezing operator $\Aop$ is essentially self-adjoint. In particular,
    if $\kappa > 2$, $\Aop$ is bounded from below by $-2|\xi|\beta^{kl}_{N-1}\kappa/(\kappa-2)$.
\end{proposition}
\begin{proof}
    By assumption, there exists $N \in \nnum$ and $\kappa \geq 2$ such that $f(n) \geq \kappa |\xi| \beta^{kl}_n$ for all $n \geq N$.
    Using~\cref{lem:rel-bounded}, there exists $b \in \rnum$ such that 
    \begin{equation}
       \norm{\xi (a^\dagger)^k a^l \psi+\xi^\ast (a^\dagger)^l a^k \psi} \leq \frac{2}{\kappa} \norm{f(a^\dagger a)\psi}+ b \norm{\psi} 
    \end{equation}
    for all $\psi \in \domain_0$.
    Thus, 
    \begin{equation}
        \Aop = \xi (a^\dagger)^k a^l+f(a^\dagger a) + \xi^\ast (a^\dagger)^l a^k
    \end{equation}
    is a relatively bounded perturbation of $f(a^\dagger a)$, with relative bound $\frac{2}{\kappa} \leq 1$.
    But by~\cref{lem:f-esssa}, $f(a^\dagger a)$ is essentially self-adjoint, and from Wüst's Theorem it follows that $\Aop$ is essentially self-adjoint.
    If $\kappa > 2$, the Kato--Rellich theorem implies
    \begin{equation}
       \braket{\psi,\Aop\psi} \geq -\frac{2|\xi| \beta^{kl}_{N-1}}{1-2/\kappa} \|\psi\|^2\, 
    \end{equation} 
    for all $\psi\in\domain_0$.
\end{proof}

We can again simplify the statement with more explicit conditions on $f(n)$ by using~\cref{lem:beta-bounds}:
\begin{corollary}\label{cor:faster}
    Let $k > l$, $\xi \in \cnum$ and $f: \nnum \to \rnum_+$. 
    Assume that there exists $\kappa > 2$ and $N \in \nnum$ such that for all $n \geq N$ $f(n) \geq \kappa |\xi| n^{(k+l)/2}$.
    Then the $(k,l)$-squeezing operator $\Aop$ is essentially self-adjoint and bounded from below.
\end{corollary}
\begin{proof}
    As $\kappa/2 > 1$, by~\cref{lem:beta-bounds} there exists $2< \tilde{\kappa} < \kappa$ and $N_1 \in \nnum$ such that $ \kappa n^{(k+l)/2} \geq \tilde{\kappa} \beta^{kl}_n$ for $n \geq N_1$, and thus the assumptions of~\cref{prop:rel-bounded} are fulfilled with $N_2 = \max(N,N_1)$.
\end{proof}

\begin{remark}
    Note that, while \cref{prop:rel-bounded} covers the critical case $\kappa = 2$ thanks to Wüst's Theorem, we cannot treat $\kappa = 2$ in \cref{cor:faster}, as $n^{(n+l)/2} < \beta^{kl}_n$. In such a case, one needs to take a closer look at the ratio $f(n)/(|\xi| \beta^{kl}_n)$ and check the slightly weaker condition in \cref{prop:rel-bounded}.
\end{remark}

To summarize, we showed two sufficient conditions under which the $(k,l)$-squeezing operator $\Aop$ is essentially self-adjoint:
\begin{enumerate}
 \item[(i)] If $l = 0$ and $k \leq 2$, $\Aop$ is essentially self-adjoint if $f(n)$ is eventually smaller than $C n^{k/2}$ for some constant $C>0$ (\cref{cor:slower}).
   \item[(ii)] Besides, for all $k,l$, $\Aop$ is essentially self-adjoint if $f(n)$ is eventually larger than $\kappa |\xi| n^{(k+l)/2}$ for some constant $\kappa>2$ (\cref{cor:faster}).
\end{enumerate}
We are thus left with analyzing what happens when $k+l\geq3$ and the squeezing term in the expression of $\Aop$ is the dominant one. As we will see in \cref{sec:deficiency}, this regime requires a more refined mathematical treatment.

\section{Squeezing operators with multiple self-adjoint extensions}
\label{sec:deficiency}

We now turn to the cases in which the $(k,l)$-squeezing operator $\Aop$ is not necessarily essentially self-adjoint.
To this purpose, we shall apply von Neumann's theory of self-adjoint extensions (see e.g.~\cite[Theorem 2.26]{teschl-mathematicalmethodsquantum-2009} or~\cite[Theorem X.2]{reed-mmmp2-fourier-1975}) to $\Aop$.
We will first show that $\Aop$ has equal deficiency indices, and thus admits self-adjoint extensions for arbitrary choices of parameters; we will then proceed to determine them in the cases not covered by the results of \cref{sec:ess-sa}, under some technical assumptions concerning the existence of asymptotic expansions.

\subsection{General properties of deficiency spaces of squeezing operators}
\label{sec:conjugation}
Recall that the \textit{deficiency subspaces} $K_\pm$ of $\Aop$ are defined as follows:
\begin{equation}
    \label{eq:def-deficiency}
    K_\pm = \Ran(\Aop \pm i)^\perp = \Ker(\Aop^\ast \mp i)\, , 
\end{equation}
with their dimensions 
\begin{equation}
    n_\pm = \dim K_\pm \, 
\end{equation}
being the \textit{deficiency indices} of $\Aop$. 
Then, according to von Neumann's theorem, $\Aop$ has self-adjoint extensions if and only if $n_+=n_-$; in this case, all self-adjoint extensions are in a one-to-one correspondence with unitary operators $U: K_+ \to K_-$. In particular, $\Aop$ is essentially self-adjoint iff $n_\pm=0$, and admits an $n_+$-dimensional vector space of self-adjoint extensions otherwise.

\begin{proposition}
\label{prop:conjugation}
    Let $k> l \in \nnum$, $\xi \in \cnum$, $f: \nnum \to \rnum_+$, and $\Aop$ be a $(k,l)$-squeezing operator.
    Then $\Aop$ has equal deficiency indices and thus admits self-adjoint extensions.
\end{proposition}
\begin{proof}
    Without loss of generality, we can set $|\xi| = 1$ by rescaling $f(n)$ accordingly, and write $\xi = \e^{\iu \theta}$ with $\theta \in \rnum$.
    Furthermore, let $\Delta = k-l$.

We will construct a conjugation $C:\hilbert\rightarrow\hilbert$, i.e. a norm-preserving antilinear operator with $C^2=I$, such that
\begin{equation}    
C\domain_0\subseteq\domain_0\qquad\text{and}\qquad\Aop C = C \Aop.
\end{equation}
This will imply the claimed result by virtue of a theorem by von Neumann \cite[Theorem X.3]{reed-mmmp2-fourier-1975}.

To this end, recall that any $\psi \in \hilbert$ can be uniquely represented as 
\begin{equation}
    \psi = \sum_{n = 0}^\infty c_n \phi_n
\end{equation}
with complex coefficients $(c_n)_{n\in\mathbb{N}}$.
We rearrange the sum of the terms with index $n\geq l$ by noticing that every such $n$ can be uniquely written as $n = n_0 + r(k-l) = n_0 +r \Delta$ for some $l \leq n_0 < k$ and $r \in \nnum$.
Hence, we can write
\begin{align}
\label{eq:psi_repr}
    \psi &= \sum_{n = 0}^{l-1} c_n \phi_n + \sum_{n_0 = l}^{k-1} \sum_{ r = 0}^\infty c_{n_0 + r\Delta} \phi_{n_0+r\Delta} \nonumber\\&= \sum_{n = 0}^{l-1} c_n \phi_n + \sum_{n_0 = l}^{k-1} \sum_{ r = 0}^\infty c^{(n_0)}_r \phi_{n_0+r\Delta} \, ,
\end{align}
where we introduced the notation $c_r^{(n_0)} = c_{n_0 +r\Delta}$ in the second equality.
We will use a similar decomposition in \cref{sec:recurrence_for_A}.
We define $C:\hilbert\rightarrow\hilbert$ as follows: for all $\psi\in\hilbert$ as in \cref{eq:psi_repr}, 
    \begin{equation}
        C \psi = \sum_{n = 0}^{l-1} c_n^\ast \phi_n + \sum_{n_0 = l}^{k-1} \sum_{r=0}^\infty \e^{2\iu r \theta}  \left(c^{(n_0)}_r\right)^\ast \phi_{n_0+r\Delta} \, .
    \end{equation}
This map is clearly antilinear, norm-preserving, and satisfies $C^2=I$ since
    \begin{align}
        C^2 \psi & = C \left( \sum_{n = 0}^{l-1} c_n^\ast \phi_n + \sum_{n_0 = l}^{k-1} \sum_{r=0}^\infty \e^{2\iu r \theta}  \left(c^{(n_0)}_r\right)^\ast \phi_{n_0+r\Delta} \right) \\
        & = \sum_{n = 0}^{l-1} c_n \phi_n + \sum_{n_0 = l}^{k-1} \sum_{r=0}^\infty \e^{2 \iu r \theta - 2\iu r \theta}c^{(n_0)}_r \phi_{n_0+r\Delta} = \psi\, ;
    \end{align}
    hence $C$ is a conjugation. Besides, clearly $C$ maps $\domain_0$ in itself. We need to show that, for every $\psi\in\domain_0$, $\Aop C\psi=C\Aop\psi$; by (anti)linearity, it suffices to prove it for each $\phi_n$.

    We distinguish the two cases $n<l$ and $n\geq l$. In the former case, we have $C\phi_n=\phi_n$ and $\beta^{kl}_n = \beta^{lk}_n= 0$, whence 
    \begin{equation}
        \Aop C \phi_n = \Aop \phi_n=  f(n) \phi_n = C \Aop  \phi_n
    \end{equation}
    immediately follows with \cref{lem:a-action}.
    In the latter case, let $l \leq n_0 < k$ and $r \in \nnum$.
    Then, using \cref{lem:a-action},
    \begin{align}
        \Aop C &  \phi_{n_0+r\Delta}  = \Aop  \e^{2 \iu r \theta} \phi_{n_0+r\Delta} \\
        & =  \e^{2 \iu r \theta} f(n_0+r\Delta)\phi_{n_0+r\Delta} +  \e^{\iu \theta+2\iu r \theta}\beta^{kl}_n\phi_{n_0+(r+1)\Delta}  + \e^{-\iu r \theta+2\iu r \theta}\beta^{lk}_n\phi_{n_0+(r-1)\Delta} \, ,
    \end{align}
while
\begin{align}
        C \Aop &   \phi_{n_0+r\Delta} = C\left( f(n_0+r\Delta)\phi_{n_0+r\Delta}  +  \e^{\iu \theta} \beta^{kl}_n\phi_{n_0+(r+1)\Delta}  +   \e^{-\iu \theta}\beta^{lk}_n\phi_{n_0+(r-1)\Delta}\right) \\
        & =   \e^{2 \iu r \theta} f(n_0+r\Delta)\phi_{n_0+r\Delta}  +  \e^{-\iu \theta+2\iu (r+1) \theta}\beta^{kl}_n\phi_{n_0+(r+1)\Delta}  +  \e^{\iu r \theta+2\iu (r-1) \theta}\beta^{lk}_n\phi_{n_0+(r-1)\Delta} \, .
    \end{align}
    Hence, $\Aop C  \phi_{n_0+r\Delta} = C\Aop  \phi_{n_0+r\Delta} $, and therefore $\Aop C \psi= C \Aop \psi$ for all $\psi \in \domain_0$.
    We showed that $C$ is a conjugation which commutes with $\Aop$, whence the claim follows.
\end{proof}

While this proposition guarantees that $\Aop$ always admits self-adjoint extensions, it does not determine whether such an extension is unique. Outside the cases treated in \cref{sec:ess-sa}, this question remains open—remarkably, even for the pure higher-order squeezing operators introduced in \cref{sec:intro}, cf. \cref{eq:k-squeeze}. In the remainder of this section, we address this problem under suitable assumptions on the asymptotic behavior of $f(n)$.

\begin{remark}\label{rem:gorska}
We remark that the case $l = 0$ and $f(n) = 0$, was already considered in Ref.~\cite{gorska-squeezingarbitraryorder-2014}, whose notation we will closely follow. However, the argument presented there contains a gap that undermines the validity of the result (cf. \cref{rem:gorska2}). Our proof addresses this issue and establishes the statement on a more general footing.
\end{remark}

\subsection{Asymptotics of linear recurrence relations}
\label{sec:recurrence}

We begin by recalling some results on linear recurrence relations which we will need for our purposes. We refer to Refs.~\cite{elaydi-introductiondifferenceequations-2005,kelley-differenceequationsintroduction-2001,mariconda-discretecalculus-2016} for a more detailed overview on the subject. Here we will focus on second-order linear recurrence relations:
\begin{equation}
    \label{eq:recurrence-birkhoff}
    d_{r+2} + a(r)d_{r+1} + b(r)d_r = 0\, \quad \forall r \in \nnum\, , 
\end{equation}
with coefficients $a(r)$ and $b(r)$. This can be seen as an equation in the vector space $\mathbb{C}^{\mathbb{N}}$ of complex sequences. 
Assuming $a(r) \neq 0$ and $b(r) \neq 0$ for all $r \in \nnum$,~\cref{eq:recurrence-birkhoff} has a unique solution for every choice of initial values $d_0$ and $d_1$~\cite[Theorem 2.7]{elaydi-introductiondifferenceequations-2005}.
Thus the solution space of~\cref{eq:recurrence-birkhoff} is two dimensional, and each solution can be written as a linear combination of two linearly independent solutions of~\cref{eq:recurrence-birkhoff}~\cite[Theorem 2.21]{elaydi-introductiondifferenceequations-2005}.

We will need to characterize the \textit{asymptotic} behavior of the solutions of~\cref{eq:recurrence-birkhoff} for $r\to\infty$. To this end, assume that the coefficients $a(r)$ and $b(r)$ converge to some finite values 
\begin{equation}\label{eq:limiting_coefficients}
    a_0 = \lim_{r \to \infty} a(r), \qquad b_0 = \lim_{r\to \infty} b(r) \, , 
\end{equation}
and consider the limiting recurrence relation
\begin{equation}\label{eq:limiting_relation}
        d_{r+2} + a_0d_{r+1} + b_0d_r = 0\, ,     
\end{equation}
which heuristically approximates the original relation~\eqref{eq:recurrence-birkhoff} for large $r$. Then, one can to determine the roots $\rho_\pm$ of the characteristic equation of the limiting recurrence relation,
\begin{equation}
    \label{eq:characteristic}
    \rho^2 + a_0 \rho + b_0 \, .
\end{equation}
If $|\rho_+| \neq |\rho_-|$, one can apply well-known results by Poincaré and Perron (see e.g.~\cite[Theorems 8.9--8.11]{elaydi-introductiondifferenceequations-2005}) to determine the asymptotics of the solutions of~\cref{eq:recurrence-birkhoff}. However, for our purposes, we need to consider the case $|\rho_+| = |\rho_-|$. Results in this regime were found by Birkhoff and Trjitzinsky~\cite{birkhoff-formaltheoryirregular-1930,birkhoff-analytictheorysingular-1933}, building upon the work of Adams~\cite{adams-irregularcaseslinear-1928}. Before introducing such results, let us recall some notation, cf.~\cite[722]{flajolet-analyticcombinatorics-2013}:

\begin{definition}\label{def:o}
    Let $f,g: \nnum \to \cnum$. We write
    \begin{enumerate}[(i)]
        \item $f(r) = O(g(r))$, if there exists $N \in \nnum$ and $C \in \rnum$ such that, for all $r \geq N$, $|f(r)| \leq C |g(r)|$.
        \item $f(r) = o(g(r))$, if $\lim_{r \to \infty} |f(r)|/|g(r)| = 0$.
    \end{enumerate}
\end{definition}
Next, we recall asymptotic expansions~\cite[724]{flajolet-analyticcombinatorics-2013}:
\begin{definition}
    \label{def:asymptotic-expansion}
    Let $f: \nnum \to \cnum$, and $(g_s)_{s \in \nnum}$ be a family of functions $g_s: \nnum \to \cnum$ satisfying $g_{s+1} = o(g_s)$ for all $s \in \nnum$.
    We say that $f$ admits the \textit{asymptotic expansion} 
    \begin{equation}\label{eq:asymptotics}
        f(r) \sim \sum_{s=0}^{\infty}\lambda_s g_s(r) \, 
    \end{equation}
    if there exist $(\lambda_s)_{s\in\nnum}\in \cnum$ such that, for all $S \in \nnum$,
    \begin{equation}
        f(r) = \sum_{s = 0}^{S} \lambda_s g_s(r) + O(g_{s+1}(r)) \, .
    \end{equation}
\end{definition}
    Importantly, $\sum_{s=1}^{\infty}\lambda_s g_s(r)$ is a \textit{formal} series expansion, in the sense that the series need not converge at any point $r$.
       
Formal series play a central role in the analysis of recurrence relations and differential equations.
When $a(r)$ and $b(r)$ admit asymptotic expansions in a given family of functions $(g_s)_{s\in\mathbb{N}}$, one often postulates a series ansatz for the solution, usually in the form
\begin{equation}
   \e^{Q(r)}\sum_{s=0}^\infty c_sg_s(r)
\end{equation}
for some function $Q(r)$, then substitutes it into the equation and replaces $a(r)$ and $b(r)$ with their asymptotic expansions, matching coefficients of the same order in $g_s$ so that the relation is satisfied at each order. This procedure yields what is called a \emph{formal series solution}, whose coefficients are determined recursively from the equation. A key question is whether such a formal series captures the asymptotic behavior of a \textit{genuine} solution of the original problem, that is, it serves as an asymptotic expansion for an actual solution. 

The Birkhoff--Trjitzinsky theory~\cite{birkhoff-analytictheorysingular-1933} addresses this question for a broad class of recurrence relations, providing precise conditions under which formal series solutions indeed correspond to asymptotic expansions of actual solutions. In our setting, the Birkhoff--Trjitzinsky theorem ensures that, whenever the coefficients \(a(r)\) and \(b(r)\) have asymptotic expansions in powers of \(r^{-1/\omega}\) for some integer \(\omega\), the equation admits two linearly independent formal series solutions that also serve as the asymptotic expansions of two genuine independent solutions. We refer to \cref{sec:appendix} for the precise statement, cf. \cref{thm:birkhoff-general}.

For our purposes, we will focus on the case \(\omega = 2\) (that is, asymptotic expansions in powers of \(r^{-1/2}\)), in which the form of our solutions simplifies considerably. While this specialization is, in principle, a direct consequence of the general theorem, we have not found it worked out explicitly in the literature (but see \cite{wong-asymptoticexpansionssecondorder-1992} for the case $\omega=1$), and therefore include a complete proof in \cref{sec:appendix}:

\begin{proposition}
    \label{prop:birkhoff}
    Consider a recurrence relation in the form~\eqref{eq:recurrence-birkhoff}. We assume the following:
\begin{itemize}
    \item[(i)] The coefficient functions $a(r)$ and $b(r)$ admit asymptotic expansions in powers of $r^{-1/2}$:
    \begin{equation}
        a(r) \sim \sum_{s=0}^{\infty} a_s r^{-s/2}\, ,  \quad b(r) \sim \sum_{s=0}^{\infty} b_s r^{-s/2}
    \end{equation}
    with $b_0 \neq 0$;
    \item[(ii)] The two solutions $\rho_\pm$  of the associated characteristic equation~\eqref{eq:characteristic} satisfy $\rho_+ \neq \rho_-$.
\end{itemize}
    Then there exists two linearly independent solutions $(d^\pm_r)_{r\in\nnum}$ of~\cref{eq:recurrence-birkhoff} with asymptotic expansions
    \begin{equation}
        \label{eq:birkhoff-solution}
        d^\pm_r \sim \rho^r_\pm \e^{\Omega_\pm r^{1/2}} r^{\alpha_\pm} \sum_{s=0}^{\infty} C^\pm_s r^{-s/2} \, , 
    \end{equation}
    where $\alpha_\pm$ and $\Omega_\pm$ are given by
    \begin{align}
        \label{eq:birkhoff-omega}\Omega_\pm & = -\frac{a_1 \rho_\pm+b_1}{a_0 \rho_\pm/2+\rho_\pm^2}\, , \\
        \label{eq:birkhoff-alpha}\alpha_\pm & = \frac{a_2 \rho_\pm+b_2}{a_0 \rho_\pm + 2 b_0} - \frac{\Omega_\pm^2(\rho_\pm/2+ a_0/8)+ \Omega_\pm a_1/2}{2\rho_\pm+a_0 } .
    \end{align}
\end{proposition}

\cref{prop:birkhoff} will serve as a central tool in our analysis of the self-adjoint extensions of $\Aop$.

\subsection{A recurrence relation for \texorpdfstring{$\Aop$}{A}}
\label{sec:recurrence_for_A}

We now return to our main quest: determining the self-adjoint extensions of the $(k,l)$-squeezing operator $\Aop$ for $k+l \geq 3$ in cases not covered by the results of \cref{sec:ess-sa}. To this end, we will explicitly determine the deficiency spaces $K_\pm$ (cf. \cref{eq:def-deficiency}) by making use of the properties of recurrence relations recalled in~\cref{sec:recurrence}. Without loss of generality, we shall henceforth assume that the parameter $\xi$ in the expression of $\Aop$ satisfies $|\xi|=1$; this can always be ensured by dividing $f(n)$ by $|\xi|$ and redefining $f(n)$ accordingly. As such, we will fix $\xi=\e^{\iu\theta}$ for some $\theta\in\rnum$.

To begin with, we represent any vector $\psi^\pm \in K_\pm$ in the basis of Fock states:
\begin{equation}
    \label{eq:psipm-expansion}
    \psi^\pm = \sum_{n=0}^{\infty} c^\pm_n \phi_n \, , \qquad c^\pm_n = \braket{\phi_n,\psi^\pm},
\end{equation}
with $(c^\pm_n)_{n\in\mathbb{N}}\in\ell^2(\mathbb{C})\subset\cnum^\nnum$.

We now show that enforcing the condition $\psi^\pm\in K_\pm$ translates into a recurrence relation on the coefficients $c^\pm_n$:
\begin{lemma}
    \label{lem:kpm-recurrence}
    Let $k>l \in \nnum$, $\theta \in \rnum$, and $\xi = \e^{\iu \theta}$.
    Let $\psi^\pm \in K_\pm$. Then its expansion coefficients $(c^\pm_n)_{n\in\mathbb{N}}$ in the Fock states, cf.~\cref{eq:psipm-expansion}, obey the following recurrence relation:
    \begin{equation}
        \label{eq:kpm-recurrence}
        \e^{-\iu \theta} \beta^{kl}_n c^\pm_{n+(k-l)}+\e^{\iu \theta} \beta^{lk}_n c^\pm_{n-(k-l)} +(f(n)\mp \iu) c^\pm_n = 0\, .
    \end{equation}
\end{lemma}
\begin{remark}\label{rem:abuse2}
    Analogously as in \cref{lem:a-action} (see \cref{rem:abuse}), here we are using a slight abuse of notation to simplify the presentation: \cref{eq:kpm-recurrence} involves coefficients $c_n^\pm$ with possibly negative values of the index $n$ up to $-(k-l)$. However, those terms are always multiplied by vanishing coefficients, and thus do not affect our discussion. An explicit version of \cref{eq:kpm-recurrence} where the various cases are presented separately is reported in the proof of \cref{lem:rec-delta-solutions}, cf. \cref{proofeq:recurrence-ll,proofeq:recurrence-geqk}. 
\end{remark}
\begin{proof}
    Let $\psi \in K_\pm$.
    Then, for all $\varphi \in \domain_0$, 
    \begin{equation}
        \braket{(\Aop\pm \iu)\varphi,\psi^\pm } = 0 \, .
    \end{equation}
    By linearity, this is equivalent to 
    \begin{equation}
        \braket{(\Aop \pm \iu)\phi_n,\psi^\pm} = 0 \qquad \forall n \in \nnum\, .
    \end{equation}
    Using~\cref{lem:a-action} with $\xi=\e^{\iu\theta}$, we obtain
    \begin{equation}
        \braket{\e^{+\iu \theta}\beta^{kl}_n \phi_{n+(k-l)}+f(n)\phi_n + \e^{-\iu \theta} \beta^{lk}_n \phi_{n-(k-l)} \pm \iu \phi_n, \psi^\pm} = 0\, .
        \end{equation}
    Using~\cref{eq:psipm-expansion}, for $n \geq k-l$ this is equivalent to
    \begin{equation}
        \e^{-\iu \theta}\beta^{kl}_n c^\pm_{n+(k-l)}+\e^{+\iu \theta} \beta^{lk}_n c^\pm_{n-(k-l)} +(f(n)\mp \iu) c^\pm_n = 0\, ,
    \end{equation}
    whence we obtain the claim.
    For $n < k-l$, we have $\beta^{lk}_{n} = 0$, and therefore, by using~\cref{eq:psipm-expansion}
    \begin{equation}
            \e^{-\iu \theta}\beta^{kl}_n c^\pm_{n+(k-l)} +(f(n)\mp \iu) c^\pm_n = 0\, ,       
    \end{equation}
    which is also equivalent to the claimed recurrence relation.
\end{proof}
Consequently, all elements of $K^\pm$ correspond to solutions $(c^\pm_n)_{n\in\mathbb{N}}\in\cnum^{\nnum}$ of the recurrence relation~\eqref{eq:kpm-recurrence} satisfying the \textit{additional} constraint $(c^\pm_n)_{n\in\mathbb{N}}\in\ell^2(\mathbb{C})$, which ensures that the corresponding series in~\cref{eq:psipm-expansion} defines an element of $\hilbert$. We will indeed show the following result:

\begin{proposition}\label{prop:conv-birkhoff}
    Let $k>l \in \nnum$ with $k+l\geq3$, and $\theta \in \rnum$. Assume that $f(n) / \beta^{kl}_n$ admits an asymptotic expansion in powers of $n^{-1/2}$:
    \begin{equation}
        \frac{f(n)}{\beta^{kl}_n} \sim \kappa + \frac{L_1}{n^{1/2}} + \frac{L_2}{n} + \dots
    \end{equation}   
    with $\kappa < 2$.
    Then the recurrence relation
    \begin{equation}
        \label{eq:kpm-recurrence-bis}
        \e^{-\iu \theta} \beta^{kl}_n c^\pm_{n+(k-l)}+\e^{\iu \theta} \beta^{lk}_n c^\pm_{n-(k-l)} +(f(n)\mp \iu) c^\pm_n = 0\, 
        \end{equation}
        admits $k-l$ linearly independent solutions in $\ell^2(\mathbb{C})$.
\end{proposition}
The remainder of this section will be devoted to proving \cref{prop:conv-birkhoff}. We will proceed as follows:
\begin{enumerate}
    \item First, we evaluate a linearly independent family of solutions of~\cref{eq:kpm-recurrence};
    \item Secondly, we will check whether these solutions are in $\ell^2(\cnum)$---and thus, correspond via \cref{eq:psipm-expansion} to proper elements of the deficiency subspaces of $\Aop$---by employing the results on the asymptotics recalled in~\cref{sec:recurrence}.
\end{enumerate}
We refer to \cref{sec:deficiency_sub} for the consequent characterization of all self-adjoint extensions of $\Aop$.

\subsubsection{Step 1: finding a linearly independent set of solutions}

We start with the following result:

\begin{lemma}
    \label{lem:rec-delta-solutions}
    Let $k>l \in \nnum$ and $\theta \in \rnum$.
    Then the recurrence relation~\eqref{eq:kpm-recurrence} has $k-l$ linearly independent solutions in $\cnum^\nnum$.
\end{lemma}

\begin{proof}
    By~\cref{lem:beta-sym}, $\beta^{kl}_n = 0$ for $n < l$ and $\beta^{kl}_n\neq0$ for $n\geq l$; analogously, $\beta^{lk}_n = 0$ for $n < k$, and $\beta^{lk}_n\neq0$ for $n\geq k$.
    Thus, we have
    \begin{align}
        \label{proofeq:recurrence-ll} (f(n)\mp\iu) c_n^\pm & = 0, && n < l \, ; \\
        \label{proofeq:recurrence-lk}\e^{-\iu \theta} \beta^{kl}_n c^\pm_{n+(k-l)} +(f(n)\mp\iu) c^\pm_n & = 0, && l \leq n < k; \, , \\
        \label{proofeq:recurrence-geqk}\e^{-\iu \theta} \beta^{kl}_n c^\pm_{n+(k-l)}+\e^{\iu \theta} \beta^{lk}_n c^\pm_{n-(k-l)} +(f(n)\mp\iu) c^\pm_n & = 0, && n \geq k\, ,
    \end{align}
    with all terms $\beta^{kl}_n$ and $\beta^{lk}_n$ above being nonzero. The linear system~\eqref{proofeq:recurrence-ll} is decoupled from the system corresponding to \cref{proofeq:recurrence-lk,proofeq:recurrence-geqk}, as the former only involves the coefficients $(c^\pm_n)_{n<l}$ and the latter only involve the coefficients $(c^\pm_n)_{n\geq l}$.  Besides, as $f(n) \in \rnum$,~\cref{proofeq:recurrence-ll} uniquely fixes all $(c^\pm_n)_{n<l}$ to zero. So the dimensionality of the space of solutions of the recurrence relation~\eqref{eq:kpm-recurrence} coincides with the one of~\eqref{proofeq:recurrence-lk}--\eqref{proofeq:recurrence-geqk}.
           
    Now, \cref{proofeq:recurrence-geqk} is a linear recurrence relation of order $2(k-l)$ in $(c^\pm_n)_{n\geq l}$, thus admits $2(k-l)$ linearly independent solutions~\cite[Theorem 2.18]{elaydi-introductiondifferenceequations-2005}. At the same time, \cref{proofeq:recurrence-lk} imposes $k-l$ additional constraints on $(c^\pm_n)_{n\geq l}$, namely
    \begin{equation}
        c^\pm_{n+(k-l)} = - \frac{(f(n) \mp\iu) \e^{\iu \theta}}{\beta^{kl}_{n}} c^\pm_{n} \qquad l \leq n < k \, ,
    \end{equation}
    which reduces the dimension of the space of solutions to $k-l$, thus implying the claim.
  \end{proof}

We then need to construct $k-l$ linearly independent solutions of the recurrence relation~\eqref{eq:kpm-recurrence}. To this end, we note that the recurrence relation is \textit{tridiagonal}: for each $n\in\mathbb{N}$,~\cref{eq:kpm-recurrence} only involves the three entries $c^\pm_{n-(k-l)},c^\pm_n,c^\pm_{n+(k-l)}$ of the sequence to be determined. Consequently, it decomposes into $k-l$ decoupled recurrence relations of second order. One can thus construct $k-l$ linearly independent solutions of \cref{eq:kpm-recurrence} by solving each of these $k-l$ second-order relations, and ``completing'' them with zeros. This is the content of the following lemma:

\begin{lemma}
\label{lem:rec-decouple}
    Let $k > l \in \nnum$, $\theta \in \rnum$, $n_0\in\mathbb{N}$ such that $l \leq n_0 < k$, and assume that $(\tcrpm_r)_{r \in \nnum}$ is a solution of the second-order recurrence relation
    \begin{equation}
    \label{eq:kpm-recurrence-r}
        \e^{-\iu \theta} \beta^{kl}_{n_0+r(k-l)} \tcrpm_{r+1}+ \e^{+\iu \theta} \beta^{lk}_{n_0+r(k-l)}\tcrpm_{r-1} + (f(n_0+r(k-l))\mp \iu)\tcrpm_{r} = 0
    \end{equation}
    for all $r \in \nnum$. Then the sequence $(\cnpm_n)_{n \in \nnum}$ given by
    \begin{equation}
    \label{eq:cnpm-construction}
        \cnpm_n = \begin{cases}
            \tcrpm_r, & n = n_0+r(k-l);\\
            0,&\text{otherwise},
        \end{cases}
    \end{equation}
    obeys the recurrence relation~\eqref{eq:kpm-recurrence}. Furthermore, sequences $(c^{(n_0,\pm)}_n)_{n\in\mathbb{N}}$ corresponding to distinct values of $n_0$ are linearly independent.
\end{lemma}
\begin{proof}
    Let $n \in \nnum$. If $n < l$, then $\cnpm_n = 0$ and \cref{eq:kpm-recurrence} is fulfilled, so let $n \geq l$. As $k>l$, there exists $l \leq m_0 < k$ and $r \in \nnum$ such that $n = m_0 +r(k-l)$.
    If $m_0 \neq n_0$, we get
    \begin{equation}
        \cnpm_n = \cnpm_{n+k-l} = \cnpm_{n-(k-l)} = 0
    \end{equation}
    and \cref{eq:kpm-recurrence} holds.
    Thus only the case $n_0 = m_0$ remains.
    Then, $\cnpm_{n} = \tcrpm_r$, $\cnpm_{n+k-l} = \tcrpm_{r+1}$ and $\cnpm_{n-(k-l)} = \tcrpm_{r-1}$, and \cref{eq:kpm-recurrence} becomes
    \begin{align}
         \e^{-\iu \theta} \beta^{kl}_n \cnpm_{n+(k-l)}+\e^{\iu \theta} \beta^{lk}_n \cnpm_{n-(k-l)} +(f(n)\mp \iu) \cnpm_n & = 0 \\
          \Leftrightarrow \quad \e^{-\iu \theta} \beta^{kl}_{n_0+r(k-l)} \tcrpm_{r+1}+\e^{\iu \theta} \beta^{lk}_{n_0+r(k-l)}\tcrpm_{r-1} 
          +(f(n_0+r(k-l))\mp \iu) \tcrpm_{r} & = 0\, .
    \end{align}
    But this is fulfilled by assumption.

    Finally, linear independence of solutions competing to distinct values of $n_0$ follows immediately from the fact that, by their definition~\eqref{eq:cnpm-construction}, for each value of the index $n$ only at most one solution---the one with $n=n_0+r(k-l)$ for some integer $r$---will have a nonzero $n$th entry.
\end{proof}

We have now decoupled the linear recurrence relation \cref{eq:kpm-recurrence} of order $2(k-l)$ in $n \in \nnum$ into $k-l$ second-order recurrence relations \cref{eq:kpm-recurrence-r} in $r \in \nnum$.
We will treat these recurrence relations separately and, in preparation for the second step of our proof, we slightly simplify them in the next lemma.
\begin{lemma}
\label{lem:rec-dr-def}
     Let $k > l \in \nnum$, $\theta \in \rnum$, $l \leq n_0 < k$.
     Set $\Delta = k-l$.
     Let $(\dr_r)_{r \in \nnum}$ be a solution of the recurrence relation
     \begin{equation}
        \label{eq:rec-dn0r}
         \beta^{kl}_{n_0+r\Delta} \dr_{r+1} - (1+\iu f(n_0+r\Delta)) \dr_r - \beta^{lk}_{n_0 + r\Delta} \dr_{r-1} = 0\, .
     \end{equation}
     Then the sequences $(\tcrpm_r)_{r \in \rnum}$ given by
     \begin{equation}
         \tilde{c}^{(+,n_0)} = \iu^r \e^{\iu r \theta} \dr_r \, ,  \qquad \tilde{c}^{(-,n_0)} = (-\iu)^r \e^{\iu r \theta} \bigl(\dr_r\bigr)^\ast
     \end{equation}
     obey the recurrence relation~\eqref{eq:kpm-recurrence-r}.
\end{lemma}
\begin{proof}
    In the following, for the ease of notation, we set
    \begin{equation}
        \label{proofeq:recurrence-dr-conj}
        d_r^{(n_0,+)} = d_r^{(n_0)} \,  ,\quad  d_r^{(n_0,-)} = \bigl(d_r^{(n_0)}\bigr)^\ast \, ,
    \end{equation}
    and therefore 
    \begin{equation}\label{proofeq:ctod}
        \tcrpm_r = (\pm \iu)^r \e^{\iu r \theta} \drpm_r \, .
    \end{equation}
    With $\Delta = k-l$, \cref{eq:kpm-recurrence-r} becomes
    \begin{align}
        \e^{-\iu \theta} \beta^{kl}_{n_0+r\Delta} \tcrpm_{r+1}+ \e^{+\iu \theta} \beta^{lk}_{n_0+r\Delta}\tcrpm_{r-1} + (f(n_0+r \Delta) \mp \iu)\tcrpm_{r} & = 0 \\
        \Leftrightarrow \quad (\pm \iu)^{r+1} \e^{\iu(r+1) \theta} \e^{-\iu \theta} \beta^{kl}_{n_0+r\Delta} \drpm_{r+1} + (\pm \iu)^{r-1} \e^{\iu(r-1) \theta} \e^{+\iu \theta} \beta^{lk}_{n_0+r\Delta} \drpm_{r-1} \quad & \\
        + (\pm \iu)^r \e^{\iu r \theta} (f(n_0+r \Delta) \mp \iu) \drpm_r &  = 0 \\
        \Leftrightarrow \quad (\pm \iu)^2 \beta^{kl}_{n_0+r\Delta} \drpm_{r+1} + \beta^{lk}_{n_0+r\Delta} \drpm_{r-1} + (\pm \iu) (f(n_0+r \Delta) \mp \iu) \drpm_r &  = 0 \\
        \Leftrightarrow \quad - \beta^{kl}_{n_0+r\Delta} \drpm_{r+1} + \beta^{lk}_{n_0+r\Delta} \drpm_{r-1} + (\pm \iu f(n_0+r \Delta) +1) \drpm_r &  = 0 \\ 
        \Leftrightarrow \quad  \beta^{kl}_{n_0+r\Delta} \drpm_{r+1} -  (1 \pm \iu f(n_0+r \Delta) ) \drpm_r -\beta^{lk}_{n_0+r\Delta} \drpm_{r-1} &  = 0. 
    \end{align}
    In the ($+$) case, this is fulfilled by assumption as $ d_r^{(n_0,+)} = d_r^{(n_0)}$.
    In the ($-$) case, using \cref{proofeq:recurrence-dr-conj}, we get
    \begin{equation}
        \beta^{kl}_{n_0+r\Delta} \bigl(\dr_{r+1}\bigr)^\ast -  (1 - \iu f(n_0+r \Delta) )  \bigl(\dr_{r}\bigr)^\ast -\beta^{lk}_{r+\Delta}  \bigl(\dr_{r-1}\bigr)^\ast  = 0.
    \end{equation}
    As $\beta^{kl}_n, \beta^{lk}_n$ and $f(n)$ are real for all $n \in \nnum$, this is the complex conjugate of \cref{eq:rec-dn0r} and therefore fulfilled by assumption.
\end{proof}

Summing up: we showed that the recurrence relation~\eqref{eq:kpm-recurrence} yielding the elements of the deficiency subspace of $\Aop$ admits $k-l$ linearly independent solutions (\cref{lem:rec-delta-solutions}) and that it decouples into $k-l$ recurrence relations of second order (\cref{lem:rec-decouple}), which can  then be recast in a slightly simplified form (\cref{lem:rec-dr-def}). As each of these is a linear recurrence relation of second order with nonzero coefficients, each of these equations can be successfully solved.

As previously remarked, however, we still need to check whether these solutions are in $\ell^2(\mathbb{C})$. We will check this in the next step by making use of the results in \cref{sec:recurrence}.

\subsubsection{Step 2: checking square summability}

To begin with, we note that each of the solutions $(c_n^{(n_0,+)})_{n\in\mathbb{N}}$ of the recurrence relation~\eqref{eq:kpm-recurrence} corresponding to a solution $(d^{(n_0)}_r)_{r\in\mathbb{N}}$ of the recurrence relation~\eqref{eq:rec-dn0r} satisfies
\begin{equation}
\label{eq:cn-norm-pm}
    |c_n^{(n_0,+)}| = |c_n^{(n_0,-)}| = \begin{cases}
        |d^{(n_0)}_r|, & n = n_0 + r \Delta \text{ for some } r \in \nnum; \\
        0, & \text{otherwise}.
    \end{cases}
\end{equation}
Hence, we only need to determine whether the sequences $(d_r^{(n_0)})_{r \in \nnum}$ are in $ \ell^2(\cnum)$. As these are solutions of a second-order linear recurrence relation with nonzero coefficients, we can apply the asymptotic results recalled in \cref{sec:recurrence} to this end.

To simplify the notation, from now on we drop the indices $n_0,k,l$, and assume the dependence on them implicitly. We will thus write $d_r = d_r^{(n_0)}$, and also adopt the shorthands
\begin{equation}
    \label{eq:dr-shorthand}
    \gamma_r = \beta^{lk}_{n_0+r\Delta}, \qquad \delta_r = f(n_0+r\Delta) \, ,
\end{equation}
with $\Delta = k-l$. 
By \cref{lem:beta-sym}, 
\begin{equation}
    \beta^{kl}_{n_0+r\Delta} = \beta^{lk}_{n_0+r\Delta + k-l} = \beta^{lk}_{n_0+(r+1)\Delta} = \gamma_{r+1} \, .
\end{equation}
Thus, \cref{eq:rec-dn0r} becomes
\begin{equation}
    \gamma_{r+1}d_{r+1} -(1+\iu \delta_r) d_r - \gamma_r d_{r-1} \, .
\end{equation}
Dividing by $\gamma_{r+1}$ and shifting the index $r$ by $1$, we get
\begin{equation}
    \label{eq:rec-dr-normal}
    d_{r+2} - \frac{1 + \iu \delta_{r+1}}{\gamma_{r+2}} d_{r+1} - \frac{\gamma_{r+1}}{\gamma_{r+2}} d_{r} = 0 \, ,
\end{equation}
which is indeed a linear recurrence relation of the form~\eqref{eq:recurrence-birkhoff} discussed in \cref{sec:recurrence}. 

\begin{remark}\label{rem:gorska2}
In the special case $l=0$ and $f(n)=0$, the recurrence relation \cref{eq:rec-dr-normal} for $\Aop$ reduces to Eq.\ (4.15) in \cite{gorska-squeezingarbitraryorder-2014}, as expected (cf. \cref{rem:gorska}).
The authors of \cite{gorska-squeezingarbitraryorder-2014} argue that the corresponding solutions are square summable.
However, the justification they provide relies on an incorrect inequality, specifically Eq.\ (4.16) in Proposition 4.6.
To address this, from this point onward we adopt a different strategy and provide sufficient conditions for the square summability of solutions to \cref{eq:rec-dr-normal}, which in particular cover the special case treated in \cite{gorska-squeezingarbitraryorder-2014}.
\end{remark}
 
In order to apply Birkhoff's theory (cf. \cref{sec:recurrence}), we need to determine the asymptotic expansion of the coefficients in~\cref{eq:rec-dr-normal}, and then determine whether the solutions given by~\cref{prop:birkhoff} are square summable.
To this end, we show the following result:
\begin{lemma}
    \label{lem:gammar-expansion}
    Let $\gamma_r$ be defined as in~\cref{eq:dr-shorthand}.
    Then the ratios $\frac{\gamma_{r+1}}{\gamma_{r+2}}$ and $\frac{1}{\gamma_{r+2}}$ admit asymptotic expansions in powers of $r^{-1/2}$ respectively satisfying    \begin{equation}\label{eq:beta_asymptotics}
\frac{\gamma_{r+1}}{\gamma_{r+2}}\sim1-\frac{k+l}{2r}+o(r^{-1}),
    \end{equation}
and
\begin{equation}
    \frac{1}{\gamma_{r+2}} \sim \frac{1}{(\Delta r)^{(k+l)/2}} + o(r^{-(k+l)/2}) \, .
\end{equation}
\end{lemma}
\begin{proof}
Recalling the definition~\eqref{eq:dr-shorthand} of $\gamma_r$ in terms of the coefficients $\beta^{lk}_n$ and, in turn, their definition in terms of Pochhammer symbols (\cref{def:beta}), we get
    \begin{align}
        \frac{\gamma_{r+1}}{\gamma_{r+2}}& =\frac{\beta^{lk}_{n_0+(r+1)\Delta}}{\beta^{lk}_{n_0+(r+2)\Delta}}= \left(\frac{(n_0+(r+1)\Delta-k+1,k)(n_0+(r+1)\Delta-k+1,l)}{(n_0+(r+2)\Delta-k+1,k)(n_0+(r+2)\Delta-k+1,l)}\right)^{1/2} \\
            & = \left(\frac{(n_0+(r+1)\Delta)!\, (n_0+r\Delta)!\, ((n_0+(r+2)\Delta-k)!)^2}{((n_0+(r+1)\Delta-k)!)^2(n_0+(r+2)\Delta)!\, (n_0+(r+1)\Delta)!}\right)^{1/2} \\ 
            & = \left(\prod_{i=j}^{\Delta} \frac{(n_0+(r+1)\Delta-k+j)^2}{(n_0+r\Delta +j)(n_0+(r+1)\Delta +j)}\right)^{1/2} \\
 & =\prod_{j=1}^{\Delta} \sqrt{\frac{n_0+(r+1)\Delta-k+j}{n_0+(r+1)\Delta +j}} \sqrt{\frac{n_0+r\Delta-l+j}{n_0 +r\Delta+j}} \, \\
 & =\prod_{j=1}^{\Delta} \sqrt{1-\frac{k}{n_0+(r+1)\Delta+j}}\,\sqrt{1-\frac{l}{n_0+r\Delta+j}} .\label{eq:finiteproduct}
    \end{align}
This can be extended to a function on $\mathbb{R}_+$ which is clearly real analytic in $r^{-1}$, as it is the product of finitely many real analytic functions; \textit{a fortiori}, the ratio admits an asymptotic expansion in powers of $r^{-1}$. To determine its first two coefficients, it will suffice to recall that, for every $\alpha\in\mathbb{R}$, 
\begin{equation}
    (1-x)^\alpha=1-\alpha x+o(x)\qquad(x\to0).
\end{equation}
As the product above is finite, we can expand each term in \cref{eq:finiteproduct}: as $r\to+\infty$,
    \begin{align}
        \sqrt{1-\frac{k}{n_0+(r+1)\Delta+j}}&=1-\frac{k}{2\left(n_0+(r+1)\Delta+j\right)}+o(r^{-1})=1-\frac{k}{2r\Delta}+o(r^{-1}),\\
         \sqrt{1-\frac{l}{n_0+r\Delta+j}}&=1-\frac{l}{2\left(n_0+r\Delta+j\right)}+o(r^{-1})=1-\frac{l}{2r\Delta}+o(r^{-1}),
    \end{align}
    whence
    \begin{align}\label{eq:expansion}
         \sqrt{1-\frac{k}{n_0+(r+1)\Delta+j}} \sqrt{1-\frac{l}{n_0+r\Delta+j}}&=1-\frac{k+l}{2r\Delta}+o(r^{-1}),
    \end{align}
    and finally, using \cref{eq:finiteproduct} and \cref{eq:expansion}, 
    \begin{align}
        \frac{\gamma_{r+1}}{\gamma_{r+2}}=1-\frac{k+l}{2r}+o(r^{-1}),
    \end{align}
thus proving the first claim.

We can treat $\frac{1}{\gamma_{r+2}}$ in a similar manner. Per \cref{def:beta}
\begin{align}
    \beta^{lk}_{n} & = \sqrt{(n-k+1,k)(n-k+1,l)} \\
    & = \sqrt{n^{k+l}+a_1 n^{k+l-1}+\dots + a_{k+l}} = n^{(k+l)/2}\sqrt{1+a_1 n^{-1}+\dots+a_{k+l}n^{-k-l}}\, , 
\end{align}
with $(a_j)_{j=1}^{k+l} \in \rnum$, where we used the fact that the Pochhammer symbol $(x,s)$ is a polynomial of order $s$ in $x$.
Using again the properties of the square root and of real analytic functions \cite[Prop. 1.1.15]{krantz-primerrealanalytic-1992}, we thus get:
\begin{equation}
    \frac{1}{\beta^{lk}_n} \sim n^{-(k+l)/2} \frac{1}{1+o(1)} \sim n^{-(k+l)/2} + o(n^{-(k+l)/2})\, , 
\end{equation}
and finally 
\begin{align}
    \frac{1}{\gamma_{r+2}} & = \frac{1}{\beta^{lk}_{n_0+(r+2)\Delta}} \sim  (n_0+(r+2)\Delta)^{-(k+l)/2} + o((n_0+(r+2)\Delta)^{-(k+l)/2})\\
    & \sim  \frac{1}{(\Delta r)^{(k+l)/2}} + o(r^{-(k+l)/2})\, ,
\end{align}
thus completing the proof.
\end{proof}

\begin{lemma}\label{lem:conv-birkhoff}
    Let $(d_r)_{r \in \nnum}$ be a solution of the recurrence relation
    \begin{equation}
        \label{eq:lem-dr-normalform}
         d_{r+2} - \frac{1 + \iu \delta_{r+1}}{\gamma_{r+2}} d_{r+1} - \frac{\gamma_{r+1}}{\gamma_{r+2}} d_{r} = 0
    \end{equation}
    with arbitrary initial conditions $d_0,d_1 \in \cnum$, and $\gamma_r$ and $\delta_r$ be given by~\cref{eq:dr-shorthand}.
    Assume that the ratio $\frac{\delta_{r+1}}{\gamma_{r+2}}$ admits an asymptotic expansion in powers of $r^{-1/2}$:
    \begin{equation}\label{eq:assump_asy}
        \frac{\delta_{r+1}}{\gamma_{r+2}} \sim \kappa + \frac{L_1}{r^{1/2}} + \frac{L_2}{r}+\dots
    \end{equation}
    with $\kappa < 2$. Then $(d_r)_{r \in \nnum} \in \ell^2(\cnum)$.
\end{lemma}
\begin{proof}
    As discussed, the recurrence relation~\eqref{eq:lem-dr-normalform} is of the form~\cref{eq:recurrence-birkhoff} discussed in \cref{sec:recurrence}, with 
    \begin{equation}
        a(r)  = - \frac{1 + \iu \delta_{r+1}}{\gamma_{r+2}}  \,,   \qquad 
        b(r)  = -\frac{\gamma_{r+1}}{\gamma_{r+2}} \, .
    \end{equation}
    As such, in order to determine the asymptotic behavior of $d_r$, we need to determine the asymptotic behavior of the two quantities above. By our assumption~\eqref{eq:assump_asy}, and taking into account that, by \cref{lem:gammar-expansion}, $\gamma_{r+2}^{-1}$ has an asymptotic expansion in powers of $r^{-1/2}$ with leading order $r^{-(k+l)/2}$, we have
    \begin{equation}
        a(r) \sim -\iu \kappa - \iu \frac{L_1}{r^{1/2}} - \iu \frac{L_2}{r} -\dots\, , 
    \end{equation}
    where all coefficients $L_1,L_2,\ldots$ are real as $\delta_{r+1}/\gamma_{r+2} \in \rnum$.
    Besides, by~\cref{lem:gammar-expansion},
    \begin{equation}
                b(r) \sim -1 + \frac{k+l}{2r} +o(r^{-1})\,.
    \end{equation}
    That is, both $a(r)$ and $b(r)$ are nonzero functions admitting asymptotic expansions in powers of $r^{-1/2}$ whose zeroth first and second-order coefficients read
    \begin{alignat}{3}
        a_0=-\iu\kappa,&\qquad a_1=-\iu L_1,& \qquad a_2 = -\iu L_2;\\
        b_0=-1,&\qquad b_1 = 0,& \qquad b_2=\frac{k+l}{2}.
    \end{alignat}
   In particular, the characteristic roots of the associated limiting equation, i.e. the solutions of $\rho^2+a_0\rho+b_0=0$, are
    \begin{equation}
        \rho_\pm = \frac{\iu \kappa}{2} \pm \sqrt{1-\frac{\kappa^2}{4}} \, =\e^{\iu\theta_\pm},\qquad \theta_\pm = \arctan\left(\pm \kappa/\sqrt{4-\kappa^2}\right)\, .
    \end{equation}
    As $\kappa<2$, these roots are unit numbers with $\rho_+\neq\rho_-$, so that \cref{prop:birkhoff} applies: \cref{eq:lem-dr-normalform} admits a pair of linearly independent solutions $(d^\pm_r)_{r\in\mathbb{N}}$ admitting the following asymptotic expansions:
     \begin{equation}
        d^\pm_r \sim \rho^r_\pm \e^{\Omega_\pm \sqrt{r}} r^{\alpha_\pm} \sum_{s=0}^{\infty} \frac{C^\pm_s}{r^{s/2}} \, ,
    \end{equation}
    with $\Omega_\pm$ and $\alpha_\pm$ respectively given by \cref{eq:birkhoff-omega,eq:birkhoff-alpha}. Plugging $\rho_\pm$ and $b_1 = 0$ into \cref{eq:birkhoff-omega} gives
    \begin{equation}
        \Omega_\pm =-\frac{a_1}{a_0/2+\rho_\pm}= \frac{\iu L_1}{-\iu \kappa/2+ \iu \kappa/2 \pm \sqrt{1-\kappa^2/4}} =   \pm \frac{\iu L_1}{\sqrt{1-\kappa^2/4}} \equiv \pm \iu \tilde{\Omega}\, ;
    \end{equation}
similarly, by \cref{eq:birkhoff-alpha},
\begin{equation}
        \alpha_\pm = -\frac{k+l}{4} \pm \iu \varphi\,,  \quad \varphi = \frac{1}{\sqrt{4-\kappa^2}}\left( \kappa\frac{ (k+l)}{4} + L_2\right) + \frac{L_1^2 \kappa}{2(4-\kappa^2)^{3/2}} \, ,
    \end{equation}
    whence
     \begin{equation}
        d^\pm_r \sim \e^{\iu r \theta_\pm} \e^{\pm \iu \varphi\log r \pm \iu \tilde{\Omega} \sqrt{r}} r^{-(k+l)/4} \sum_{s=0}^\infty \frac{C^\pm_s}{r^{s/2}} \, .
    \end{equation}  
    Recalling \cref{def:asymptotic-expansion,def:o}, this means that there exists $C^\pm>0$ and $R>0$ such that, for $r\geq R$,
\begin{equation}
        \left|d_r^\pm-C_0^\pm\e^{\iu r \theta_\pm} \e^{\pm \iu \varphi\log r \pm \iu \tilde{\Omega}_\pm \sqrt{r}} r^{-(k+l)/4}\right|\leq C^\pm r^{-1/2-(k+l)/4},
    \end{equation}
    whence, using the triangle inequality and taking squares,
     \begin{equation}
        |d_r^\pm|^2\leq 2|C_0^\pm|^2r^{-(k+l)/2}+2(C^\pm)^2 r^{-1-(k+l)/2}.
    \end{equation}
    As $k+l\geq3$, the above inequality implies $\sum_{r\geq R}|d^\pm_r|^2<\infty$ and thus $\sum_{r\in\mathbb{N}}|d^\pm_r|^2<\infty$.
    
    We proved that the second-order recurrence relation \cref{eq:rec-dr-normal} admits two linearly independent solutions which are both square summable. As all solutions of \cref{eq:rec-dr-normal} can be obtained as a linear combination of those, the proof is complete.
\end{proof}

\begin{remark}
It is instructive to point out where the above proof would fail if either $k+l<3$, or the parameter $\kappa$ in \cref{eq:assump_asy} satisfies $\kappa>2$. In the former case, taking e.g. $k+l=2$, then we would have $\alpha_\pm = -\frac{1}{2} \pm \iu \varphi$, whence 
    \begin{equation}
        \sum_{r=0}^{\infty} |d_r|^2 \sim \sum_{r=0}^\infty r^{-1},
    \end{equation}
    which diverges. In the latter case $\kappa>2$, the quantity $\sqrt{1-\kappa^2/4}$ becomes imaginary, thus $|\rho_+| >  1$ and the asymptotic expansion of $d^+_r$ contains terms that grow with $r$. In both cases, this argument does not prove that every solution of~\cref{eq:lem-dr-normalform} is not in $\ell^2(\cnum)$, as there might still be in principle a linear combination of the two solutions that is square summable. 
    However, in many such cases we can apply the results from \cref{sec:ess-sa} to show that the operator $\Aop$ is essentially self-adjoint, in which case none of the solutions of~\cref{eq:lem-dr-normalform} can be square summable.
\end{remark}

We can finally obtain the desired result for the recurrence relation~\eqref{eq:kpm-recurrence} we started from:

\begin{proof}[Proof of \cref{prop:conv-birkhoff}]
    By \cref{lem:rec-delta-solutions}, \cref{eq:kpm-recurrence-bis} admits $k-l$ linearly independent solutions; and by \cref{lem:rec-decouple,lem:rec-dr-def}, such a system of solutions can be constructed by taking into account the solutions of the $k-l$ second-order recurrence relations
    \begin{equation}
        \label{proofeq:prop-solutions-rec}
        \beta^{kl}_{n_0+r\Delta} \dr_{r+1} - (1+\iu f(n_0+r\Delta)) \dr_r - \beta^{lk}_{n_0 + r\Delta} \dr_{r-1} = 0\, .
    \end{equation}
    By assumption, \cref{eq:dr-shorthand}, and \cref{lem:beta-sym}, we obtain
    \begin{align}
        \frac{\delta_{r+1}}{\gamma_{r+2}} & =  \frac{f(n_0+(r+1)\Delta)}{\beta^{lk}_{n_0+(r+2)\Delta}} = \frac{f(n_0+(r+1)\Delta)}{\beta^{kl}_{n_0+(r+1)\Delta}} \\
        & \sim \kappa - \frac{L_1}{(n_0 + (r+1)\Delta)^{1/2}} -\frac{L_2}{n_0+(r+1)\Delta} -\dots \\&\sim \kappa - \frac{\tilde{L}_1}{r^{1/2}} - \frac{\tilde{L}_2}{r} -\dots 
    \end{align}
    with $\kappa < 2$ and all other coefficients being real.
    Hence, by \cref{lem:conv-birkhoff}, all solutions of \cref{proofeq:prop-solutions-rec} are in $\ell^2(\mathbb{C})$, thus proving the claim.
\end{proof}

    \subsection[Self-adjoint extensions of A for k+l larger than 3]{Self-adjoint extensions of \texorpdfstring{$\Aop$}{A} for \texorpdfstring{$k+l \geq 3$}{k+l larger than 3}}\label{sec:deficiency_sub}

We can now prove the main result of this section concerning the self-adjoint extension of the $(k,l)$-squeezing operator $A_{k,l}(f)$ in the case $k+l\geq3$ and $f(n)/n^{(k+l)/2}\sim \kappa<2$. Recall that, without loss of generality, we are taking the squeezing parameter $\xi$ to have unit modulus, i.e. $\xi=\e^{\iu\theta}$ for some $\theta\in\mathbb{R}$.

\begin{theorem}
    \label{thm:a-deficiency}
    Let $k>l\in\mathbb{N}$ with $k+l \geq 3$, $f: \nnum \to \rnum_+$, and $\Aop$ be the corresponding $(k,l)$-squeezing operator.
    Assume that $f(n) / \beta^{kl}_n$ admits an asymptotic expansion in powers of $n^{-1/2}$,
    \begin{equation}
        \frac{f(n)}{\beta^{kl}_n} \sim \kappa + \frac{L_1}{n^{1/2}} +\frac{L_2}{n} + \dots
    \end{equation}
    with $\kappa < 2$.
    Then $\Aop$ is not essentially self-adjoint and has deficiency indices $n_\pm=\Delta=k-l$.
    Its deficiency spaces $K_\pm=\Ran(\Aop\pm\iu)^\perp$ are given by 
    \begin{equation}
        \label{eq:a-deficiency-spaces}
        K_\pm = \mspan (\varphi^\pm_{n_0})_{l \leq n_0<k} \, , \quad \varphi^\pm_{n_0}=\sum_{r=0}^{\infty}(\pm \iu)^r \e^{\iu r \theta} d^{(n_0,\pm)}_r \phi_{n_0+r \Delta} \, , 
    \end{equation}
    where $d^{(n_0,+)}_r = d^{(n_0)}_r$, $d^{(n_0,-)}_r = \bigl(d^{(n_0)}_r\bigr)^\ast$, and $(d^{(n_0)}_r)_{r \in \nnum}$ is a solution of~\cref{eq:rec-dn0r}.
    The self-adjoint extensions of $\Aop$ are parametrized by $\Delta\times\Delta$ unitary matrices as follows: $\mathrm{U}(\Delta)\ni U\mapsto A^U_{k,l}(f)$, where $A^U_{k,l}(f):\mathcal{D}(A^U_{k,l}(f))\subset\hilbert\rightarrow\hilbert$ is given by
\begin{equation}
\domain(A_{k,l}^U(f)) =\{\psi+\varphi_U \, : \, \psi \in \domain(\overline{\Aop}), \, \varphi_U \in \domain_U\},
\end{equation}
with
\begin{equation}\label{eq:singularpart}
     \domain_U = \left\{\sum_{i=0}^{\Delta-1} a_i \varphi_{l+i}^+ -\sum_{j=0}^{\Delta-1} a_i U_{ij} \frac{\norm{\varphi_{l+i}^+}}{\norm{\varphi_{l+j}^+}} \varphi_{l+j}^- \, : \,  (a_i)_{i = 0}^{\Delta-1} \in \cnum^\Delta \right\},
\end{equation}
 where $\overline{\Aop}$ is the closure of $\Aop$, and $A_{k,l}^U(f) \psi = \Aop^\ast \psi$. Furthermore, $\{\psi_0 + \varphi_U\, :\, \psi_0 \in \domain_0,\varphi_U \in \domain_U\}$ is a core of $A_{k,l}^U(f)$.
\end{theorem}

\begin{proof}
Let $\psi^\pm$ be an element of $\hilbert$, which we expand in the Fock state basis $(\phi_n)_{n\in\mathbb{N}}$ of $\hilbert$,
 \begin{equation}
        \psi^\pm = \sum_{n = 0}^{\infty} c_n^\pm \phi_n\, .
    \end{equation}
    By \cref{lem:kpm-recurrence}, $\psi^\pm$ is an element of the deficiency subspace $K^\pm$ of $\Aop$ if and only if the coefficients $c_n$ are a solution of \cref{eq:kpm-recurrence}; by \cref{prop:conv-birkhoff}, the latter admits precisely $k-l$ linearly independent solutions in $\ell^2(\mathbb{C})$, whence both spaces $K^\pm$ have dimensions $k-l$, proving that $\Aop$ has deficiency indices equal to $k-l=\Delta$. 
    
    We can characterize the deficiency subspaces by taking $\Delta$ linearly independent solutions $(\cnpm_n)_{n \in \nnum}$ of \cref{eq:kpm-recurrence}. Those can be constructed by means of \cref{lem:rec-decouple,lem:rec-dr-def}: given $l \leq n_0 < k$,   
    \begin{equation}
        \cnpm_n = \begin{cases}
            (\pm \iu)^r \e^{\iu r \theta} d^{(n_0,\pm)}_r & n = n_0 +r \Delta \, ; \\
            0 & \text{otherwise},
        \end{cases}
    \end{equation}
    where $d^{(n_0,+)} = \dr_r$ and $d^{(n_0,-)} = \bigl(\dr_r\bigr)^\ast$ and $(\dr_r)_{r\in\nnum}$ is a solution of \cref{eq:rec-dn0r}; correspondingly, by \cref{eq:a-deficiency-spaces} we obtain an orthogonal (non-normalized) basis $\bigl(\varphi^\pm_{n_0}\bigr)_{l \leq n_0 < k}$ of the deficiency spaces $K_\pm$.

    As $\Aop$ has equal deficiency indices $n_\pm = \Delta$, by von Neumann's theory it admits a family of self-adjoint extensions $A_{k,l}^U(f)$, parametrized by a unitary matrix $U \in \mathrm{U}(\Delta)$.
    The self-adjoint extensions are given by~\cite[82]{teschl-mathematicalmethodsquantum-2009}
    \begin{align}\label{eq:domain_alkuf}
        \domain(A_{k,l}^U(f))& = \left\{ \psi + \varphi^+-U \varphi^+ \, : \, \psi \in \domain(\overline{\Aop}), \, \varphi^+ \in K_+\right\}\,, \\\label{eq:action_alkuf}
        A_{k,l}^U(f)(\psi + \varphi^+-U \varphi^+)& = A_{k,l}^U(f)\psi+\iu\varphi_++\iu U\varphi_+,
    \end{align}
    where $U: K_+\to K_-$ is a unitary operator. Since $(\varphi^\pm_{n_0})_{l \leq n_0 < k}$ is an orthogonal basis of $K_\pm$, we can uniquely represent said operator as
    \begin{equation}
        \label{proofeq:extension-v}
        U \varphi^+ = \sum_{i,j=0}^{\Delta-1} U_{ij}  \frac{\braket{\varphi_{l+i}^+,\varphi^+}}{\norm{\varphi_{l+i}^+}}\frac{\varphi_{l+j}^-}{\norm{\varphi_{l+j}^-}} \, , 
    \end{equation}
    where $(U_{i,j})_{i,j} \in \mathrm{U}(\Delta)$. With a slight abuse of notation, we will use the symbol $U$ for this matrix as well. Furthermore, we can write every $\varphi^+ \in K_+$ as $\varphi^+ = \sum_{s=0}^{\Delta-1}a_s \varphi_{l+s}^+$ with coefficients $a_s \in \cnum$.
    Using $\norm{\varphi_i^+} = \norm{\varphi_i^-}$ for all $l \leq i < k$,~\cref{proofeq:extension-v} becomes
    \begin{equation}
        U \varphi^+ = \sum_{i,j,s=0}^{\Delta-1} U_{ij}  \frac{\braket{\varphi_{l+i}^+,a_s \varphi_{l+s}^+}}{\norm{\varphi_{l+i}^+}}\frac{\varphi_{l+j}^-}{\norm{\varphi_{l+j}^-}} 
         = \sum_{i,j=0}^{\Delta-1} a_i U_{ij} \frac{\norm{\varphi_{l+i}^+}}{\norm{\varphi_{l+j}^+}} \varphi_{l+j}^-\, ;
    \end{equation}
    hence $\{\psi^+-U \psi^+ : \psi^+ \in K_+\} = \domain_U$, and we obtain the claimed expression for $\domain(A_{k,l}^U(f))$. The action of $A_{k,l}^U(f)$ on $\domain(A_{k,l}^U(f))$ could analogously be reconstructed by \cref{eq:action_alkuf}. Furthermore, as $\Aop \subset A_{k,l}^U(f) \subset \Aop^\ast$, $A_{k,l}^U(f) \psi = \Aop^\ast \psi$ for all $\psi \in \domain(A_{k,l}^U(f))$.

    Finally, let $A_{k,l}^{U,0}(f)$ be the restriction of $A_{k,l}^U(f)$ to the domain $\{\psi_0 + \varphi_U\, :\, \psi_0 \in \domain_0,\varphi_U \in \domain_U\}$.
    As the closure of $\domain_0$ with respect to the graph norm of $A_{k,l}^{U,0}(f)$ is $\domain(\overline{\Aop})$, the operator closure of $A_{k,l}^{U,0}(f)$ is $A_{k,l}^{U}(f)$, and $A_{k,l}^{U,0}(f)$ is essentially self-adjoint, thus completing the proof.
\end{proof}

We proved that for $k+l\geq3$, under suitable conditions on the asymptotic behavior of $f(n)/\beta^{kl}_n$, the operators $\Aop$ admit a $(k-l)$-dimensional space of self-adjoint extensions. This complements the results in \cref{prop:rel-bounded}, where we found conditions on $f(n)$ under which $\Aop$ is essentially self-adjoint. Similarly to what we did in that case (cf. \cref{cor:faster}), we can simplify the conditions on $f(n)$ as follows:
\begin{corollary}
    \label{cor:a-deficiency}
    Let $k > l$ with $k + l \geq 3$, $f:\nnum \to \rnum_+$, and $\Aop$ as in~\cref{def:higher-order-squeezing}.
    Assume that $f(n) n^{-(k+l)/2}$ admits an asymptotic expansion in powers of $n^{-1/2}$:
    \begin{equation}
        f(n) n^{-(k+l)/2} \sim \kappa+\frac{L_1}{n^{1/2}} + \frac{L_2}{n} + \dots
    \end{equation}
    with $\kappa < 2$.
    Then $\Aop$ has deficiency indices $n_\pm = k-l$, and its self-adjoint extensions are parametrized as in \cref{thm:a-deficiency}.
\end{corollary}
\begin{proof}
    Per \cref{def:beta}, $\beta^{kl}_n$ can be expressed as a convergent series:
    \begin{equation}
        \beta^{kl}_n = \sqrt{(n-l+1,l)(n-l+1,k)} = n^{(k+l)/2}\left(1+\frac{a_1}{n}+\dots \right) \, ;
    \end{equation}
    hence we obtain the claim from \cref{thm:a-deficiency} and basic properties of asymptotic expansions (see e.g. \cite[Prop. 1.1.15]{krantz-primerrealanalytic-1992}):
    \begin{equation}
        \frac{f(n)}{\beta^{kl}_n} = \frac{f(n)}{n^{(k+l)/2}}\frac{1}{ \left(1+ a_1/n+\dots\right)} \sim   \kappa + \frac{\tilde{L}_1}{n^{1/2}} + \frac{\tilde{L}_2}{n}+\dots\, .\qedhere
    \end{equation}
\end{proof}
Finally, we can show that, under the same conditions as \cref{thm:a-deficiency}, $\Aop$ is not bounded from below, i.e. it has no ground state energy:
\begin{proposition}
\label{prop:no-lower-bound}
    Let $k>l$ with $k+l \geq 3$, $f: \nnum \to \rnum_+$ and  $\Aop$ as in~\cref{def:higher-order-squeezing}.
    Assume that $f(n) / \beta^{kl}_n$ admits an asymptotic expansion in powers of $n^{-1/2}$,
    \begin{equation}
        \frac{f(n)}{\beta^{kl}_n} \sim \kappa + \frac{L_1}{n^{1/2}} + \frac{L_2}{n} + \dots
    \end{equation}
    with $\kappa < 2$.
    Then $\Aop$ is unbounded below, and so are all its self-adjoint extensions $A_{k,l}^U(f)$.
\end{proposition}
\begin{proof}
    By \cref{def:asymptotic-expansion}, there exists $C_1>0$ and $N \in \nnum$ such that, for all $n \geq N$,
    \begin{equation}
        \left|\frac{f(n)}{\beta^{kl}_n} - \kappa \right| < \frac{C_1}{n^{1/2}}\, ;
    \end{equation}
    hence there exists $\kappa < \kappa_1<2$ and $\tilde{N} \in \nnum$ such that 
    \begin{equation}
        \label{proofeq:no-lower-bound-k1}
        \frac{f(n)}{\beta^{kl}_n} \leq \kappa_1
    \end{equation}
    for all $n \geq \tilde{N}$.
    Let $\tilde{R} \in \nnum$ and $l \leq n_0 \leq k-1$ such that $\tilde{N} = n_0 + \tilde{R} \Delta$, where $\Delta = k-l$.
    In the following, we use the shorthands $\gamma_r, \delta_r$ as introduced in~\cref{eq:dr-shorthand}.
    By \cref{lem:gammar-expansion}, $\frac{\gamma_r}{\gamma_{r+1}} \sim 1+o(1)$, hence for all $C_2>1$ there exists $\tilde{R} \leq R \in \nnum$ such that, for all $r \geq R$,
    \begin{equation}
        \gamma_{r+1} \leq C_2 \gamma_r\, .
    \end{equation}
    As $C_2 > 1$ was arbitrary, it follows together with \cref{proofeq:no-lower-bound-k1} that there exists $\tilde{\kappa} = \kappa_1 C_2 <2$ such that
    \begin{equation}
            \label{proofeq:no-lower-bound-k2}
        \delta_r = f(n_0+r \Delta ) \leq \kappa_1 \beta^{kl}_{n_0+r\Delta} = \kappa_1\gamma_{r+1} \leq \kappa_1 C_2 \gamma_r = \tilde{\kappa} \gamma_r 
    \end{equation}
    for all $r \geq R$.

    To prove that $\Aop$ is not bounded from below, we define the following sequence
    \begin{equation}
        \psi^{(n)} = \sum_{r=R}^{R+n-1} c^{(n)}_r \phi_{n_0+r\Delta} \,  \quad c^{(n)}_r = \frac{(-1)^r \e^{\iu r \theta}}{\sqrt{n}} \, ,
    \end{equation}
    with $\norm{\psi^{(n)}} = 1$ for all $n \in \nnum$ and show that $\lim_{n \to \infty}\braket{\psi^{(n)}, \Aop \psi^{(n)}} = -\infty$.
    Using~\cref{lem:a-action}, we obtain
    \begin{align}
        \braket{\psi^{(n)}, \Aop \psi^{(n)}} & =  \sum_{r = R}^{R+n-1} \sum_{s=R}^{R+n-1} \frac{(-1)^{r+s} \e^{\iu (s-r) \theta}}{n}\left( \braket{\phi_{n_0+r\Delta},\e^{\iu \theta} \gamma_{s+1}\phi_{n_0+(s+1)\Delta}}\right. \nonumber \\
        & \quad + \left. \braket{\phi_{n_0+r \Delta},\delta_s \phi_{n_0+s\Delta}} + \braket{\phi_{n_0+r\Delta},\e^{-\iu \theta}\gamma_s \phi_{n_0+(s-1)\Delta}}\right) \\
        & = \frac{1}{n} \sum_{r = R}^{R+n-1} \delta_r - \frac{1}{n} \sum_{r=R+1}^{R+n-1} \e^{-\iu \theta+\iu \theta } \gamma_r - \frac{1}{n} \sum_{r=R}^{R+n-2} \e^{\iu \theta -\iu \theta} \gamma_{r+1} \\
        & = \frac{1}{n} \sum_{r = R}^{R+n-1} \delta_r - \frac{2}{n} \sum_{r=R+1}^{R+n-1} \gamma_r \leq \frac{\delta_R}{n} + \frac{\tilde{\kappa}-2}{n}\sum_{r=R+1}^{R+n-1} \gamma_r \, , 
    \end{align}
    where we used $\delta_r \leq \tilde{\kappa} \gamma_r$ (see \cref{proofeq:no-lower-bound-k2}) in the last step.
    Since $k + l \geq 3$, we obtain $\lim_{n \to \infty} \gamma_{R+n+1}/n = \infty$.
    As $\tilde{\kappa} -2<0$ per assumption, we proved that $\Aop$ is unbounded below. This implies that all self-adjoint extensions $A_{k,l}^U(f)$, cf. \cref{thm:a-deficiency}, are also unbounded below, as $\domain(\Aop)\subset\domain(A^U_{k,l}(f))$ and the action of the two operators coincide on the former space.
\end{proof}

\section{Polynomial field term and examples}\label{sec:examples}

We will now come back to the setting considered in \cref{sec:intro}, and specialize our results to the case where the field term in $\Aop$ is a polynomial in the number operator $a^\dag a$. We will then analyze a selection of relevant examples.

\begin{proposition}
\label{prop:polynomial}
    Let $k>l\in\mathbb{N}$, $h \in \nnum$, $\xi  \in \cnum$, $f:\mathbb{N}\rightarrow\mathbb{R}_+$ be given by
    \begin{equation}
        f(n)=\sum_{j=0}^h a_jn^j,\qquad a_0,\ldots,a_{h}\geq0\,,
    \end{equation}
    and also assume $a_h > 0$ unless $f=0$.
    Let $\Aop$ the corresponding $(k,l)$-squeezing operator. Then the following statements hold:
    \begin{enumerate}[(i)] 
        \item \label{enum:prop-poly-leq2} If $k+l<3$, then $\Aop$ is essentially self-adjoint;
        \item \label{enum:prop-poly-geq3} If $k+l\geq3$, then:
        \begin{enumerate}
            \item[(ii-a)] \label{enum:prop-poly-geq3-ess} If $2h>k+l$, or $2h=k+l$ and $a_h>2 |\xi|$, then $\Aop$ is essentially self-adjoint and bounded from below;
            \item[(ii-b)] \label{enum:prop-poly-geq3-extension}If $2h<k+l$, or $2h=k+l$ and $a_h<2 |\xi|$, then $\Aop$ is neither essentially self-adjoint nor bounded from below, has deficiency indices $n_\pm=k-l$, and its self-adjoint extensions are parametrized as given by \cref{thm:a-deficiency}.
        \end{enumerate}
    \end{enumerate}
\end{proposition}
\begin{proof}
Let $k+l<3$, i.e. $l=0$ and $k\leq2$. If $2h\leq k$, then necessarily $h = 0$ or $h = 1$, and for all $n\geq1$, we have $n^j\leq n^{k/2}$ for all $j=0,\dots,h$, whence
\begin{equation}
    f(n)\leq \sum_{j=0}^h a_j n^{k/2}\equiv a n^{k/2},
\end{equation}
whence $\Aop$ is essentially self-adjoint by \cref{cor:slower}. Instead, if $2h >  k$, there exists $N \in \nnum$ such that $a_h n^h > 3 |\xi| n^{k/2}$ for all $n \geq N$, and thus
\begin{equation}
    f(n) \geq a_h n^h > 3 |\xi| n^{k/2}\, ;
\end{equation}
and by \cref{cor:faster} $\Aop$ is essentially self-adjoint, which concludes case \labelcref{enum:prop-poly-leq2}.

We now consider case \labelcref{enum:prop-poly-geq3}, i.e. let $k+l \geq 3$.
If $2h > k+l$, there exists $N \in \nnum$ such that $a_h n^h > 3 |\xi| n^{(k+l)/2}$ for all $n \geq N$, hence $f(n) > 3|\xi|  n^{(k+l)/2}$ and $\Aop$ is essentially self-adjoint and bounded from below by \cref{cor:faster}.
If, on the other hand, $2h \leq k+l$, $f(n) n^{-(k+l)/2}$ is a linear combination of powers of $n^{-1/2}$:
\begin{equation}
    f(n) n^{-(k+l)/2} = \sum_{j = 0}^h a_j n^{j-(k+l)/2} 
\end{equation}
Specifically, if $2h < k+l$, there is no $0$th-order (constant) contribution, and by \cref{cor:a-deficiency}, $\Aop$ admits $k-l$ self-adjoint extensions parametrized by \cref{thm:a-deficiency}.
Furthermore, by \cref{prop:no-lower-bound}, $\Aop$ is also not bounded from below in this case.

If $2h = k+l$, the essential self-adjointness crucially depends on the ratio $a_h/|\xi|$.
For $a_h > 2 |\xi|$, 
\begin{equation}
    f(n) \geq a_h n^h > 2 |\xi| n^{(k+l)/2}\, ,   
\end{equation}
and $\Aop$ is essentially self-adjoint and bounded from below by \cref{cor:faster}.
If, on the other hand, $a_h < 2 |\xi|$, the $0$th-order coefficient is $\kappa = a_h / |\xi| < 2$ in \cref{cor:a-deficiency}, and $\Aop$ is not essentially self-adjoint, admitting $k-l$ self-adjoint extensions parametrized by \cref{thm:a-deficiency}.
Applying \cref{prop:no-lower-bound}, $\Aop$ and its self-adjoint extensions are also not bounded from below.
\end{proof}

This proposition allows us to directly apply our results to a wide range of operators of physical interest.
To begin with, we consider the higher-order squeezing operator from \cref{eq:k-squeeze}:
\begin{example}[Higher-order squeezing term]\label{ex:k-squeeze}
    Let $k \geq 3$ and $\xi \in \cnum$.
    We consider again the higher-order squeezing operator $A_k = \xi (a^\dagger)^k+\xi^\ast a^k$ introduced in \cref{sec:intro}.
    As $f=0$, the conditions of Case (ii-b) in \cref{prop:polynomial} are clearly fulfilled, thus $A_k$ has deficiency indices $n_\pm = k$, and admits a family of self-adjoint expansions $A_k^U$, parametrized by $k\times k$ unitary matrices $U \in \mathrm{U}(k)$. 
    
    Let us explicitly show how the parametrization from \cref{thm:a-deficiency} simplifies in such a setting. For simplicity, again we fix $\xi=\e^{\iu\theta}$. Since $f(n) = 0$, the solution $(\dr_r)_{r \in \nnum}$ with initial condition $\dr_0 = 1$ of the recurrence relation~\eqref{eq:rec-dn0r} is a sequence of real numbers, which implies $d^{(n_0,+)}_r = d^{(n_0,-)}_r=\dr_r$ in \cref{thm:a-deficiency}. Consequently, 
    \begin{align}
       \domain(A_k^U) & =\{\psi+\varphi_U \, : \, \psi \in \domain(\overline{A_k}), \, \varphi_U \in \domain_U\}\, , \\
       \domain_U & = \left\{\sum_{i=0}^{k-1} a_i \varphi_{l+i}^+ -\sum_{j=0}^{k-1} a_i U_{ij} \frac{\norm{\varphi_{l+i}^+}}{\norm{\varphi_{l+j}^+}} \varphi_{l+j}^- \, : \,  (a_i)_{i = 0}^{k-1} \in \cnum^k \right\} \, , 
    \end{align}
    where $\varphi_{n_0}^\pm = \sum_{r = 0}^\infty (\pm \iu)^r \e^{\iu r \theta} \dr_r \phi_{n_0+rk}$ for $0 \leq n_0 < k$.
    Furthermore, $\{\psi_0 + \varphi_U\, :\, \psi_0 \in \domain_0,\varphi_U \in \domain_U\}$ is a core of $A_k^U$.
    By \cref{prop:no-lower-bound}, both $A_k$ and all its self-adjoint extensions $A_k^U$ are not bounded from below.

    Finally, we mention that two particular self-adjoint extensions of $A_k^U$, respectively corresponding to $U=I$ and $U=-I$, play a particularly important role in numerical simulations. We comment on this point in \cref{sec:conclusion}.
\end{example}

We shall now consider what happens in the presence of an additional Kerr term, i.e. $K(a^\dag)^ha^h$ for some integer $h$, cf. \cref{eq:k-squeeze-Kerr}, and $K\geq0$ being a coupling constant. Let us first show that such a term is indeed in the form $f(a^\dag a)$ considered in this paper. To this end, a direct calculation analogous to the one in the proof of \cref{lem:beta-sym} yields
\begin{equation}
    (a^\dagger)^h a^h \phi_n= (n-h+1,h)\phi_n \, ,
\end{equation}
where, as in the remainder of the paper, $(x,s)$ denotes the Pochhammer symbol, cf. \cref{eq:pochhammer}. Therefore, the $h$-order Kerr term corresponds to the choice of function 
\begin{equation}\label{eq:Kerr_function}
    f(n)=K(n-h+1,h)\, , 
\end{equation}
which is polynomial of order $h$ in $n$.
We first consider the standard (second-order) squeezing operator with a Kerr term, and then proceed to analyze the case of higher-order squeezing operators.

\begin{example}[Squeezing operator with Kerr term]\label{ex:2-squeeze-Kerr}
    Let $\xi \in \cnum$, $h \in \nnum$ and $K >0$, and consider the operator
    \begin{equation}
        A_{2,\rm Kerr} = \xi (a^\dagger)^2 + \xi^\ast a^2+ K (a^\dagger)^h a^h  \, , \quad \domain(A) = \domain_0 \, ,
    \end{equation}
This coincides with the operator $\Aop$ discussed in this paper, with $k=2$, $l=0$, and $f(n)$ as per \cref{eq:Kerr_function}.
The conditions of Case (i) in \cref{prop:polynomial} hold, so this operator is essentially self-adjoint for any value of $h$.
\end{example}

\begin{example}[Higher-order squeezing operator with Kerr term]\label{ex:k-squeeze-Kerr}
We now consider the operator from \cref{eq:k-squeeze-Kerr}:
\begin{equation}
        A_{k,\rm Kerr} = \xi (a^\dagger)^k + \xi^\ast a^k +K (a^\dagger)^h a^h\, , \quad \domain(A) = \domain_0 \, ,
\end{equation}
with $\xi \in \cnum$, $K>0$, and $k,h \in \nnum$ with $k\geq3$. Again, this coincides with the operator $\Aop$ discussed in this paper with $l=0$ and $f(n)$ as per \cref{eq:Kerr_function}. 
As $f(n)$ is a polynomial of order $h$ in $n$, we can distinguish three cases using \cref{prop:polynomial}:
\begin{enumerate}[(i)]
    \item $k>2h$. Then $A_{k,\rm Kerr}$ has deficiency indices $n_\pm = k$ and is not essentially self-adjoint.
    Furthermore, it is not bounded from below and all its self-adjoint extensions are also not bounded from below.
    \item $k < 2h$: $A_{K,\rm Kerr}$ is essentially self-adjoint and bounded from below.
    \item $k=2h$. Here, the essential self-adjointness of $A_{k,\rm Kerr}$ critically depends on the value of the ratio $K/|\xi|$. If $K/|\xi| > 2$, , the operator is essentially self-adjoint and bounded from below. If, on the other hand, $K/|\xi| < 2$, the operator has deficiency indices $k$ and is not bounded from below.
\end{enumerate}
These result confirm our intuition. The squeezing and the Kerr terms are respectively of order $k$ and $2h$ in the creation and annihilation operators. When $k<2h$, the Kerr term ``dominates'' over the higher-order squeezing term in the large-photon limit, thus it works as a regulating term which both restores essential self-adjointness and finiteness of the ground state energy. When $k>2h$,  the Kerr term just acts as a perturbation of the operator $A_k$, and will still admit infinitely many self-adjoint extensions. In the case $k=2h$, where the two terms are of equal order and none of them prevails in the large-photon limit, the deciding quantity becomes the ratio between the corresponding multiplicative constants $K$ and $|\xi|$.
\end{example}

The situation analyzed in \cref{ex:2-squeeze-Kerr,ex:k-squeeze-Kerr} for the Kerr term can be readily extended to more complicated polynomials in the number operator by means of \cref{prop:polynomial}: the asymptotic behavior of $f(n)$ for large $n$---and therefore, whether the term of highest order dominates over the squeezing term or not---will determine whether essential self-adjointness is achieved. 

To conclude, we note that it is still possible to construct examples to which our abstract results do not apply—for instance, when $f(n)$ exhibits an oscillatory behavior that precludes both the large-$n$ estimates crucial for the results in \cref{sec:ess-sa} and the existence of asymptotic expansions underlying the analysis in \cref{sec:deficiency}. For completeness, we briefly present such a technical example:

\begin{example}
\label{ex:no-result}
    We consider the operator $\Aop$ with $k > l$, $k+l \geq 3$, and 
    \begin{equation}
        f(n) = n^{(k+l)/2+1} |\sin(n \pi/2)|.
    \end{equation}
    As $\sin(n \pi/2) = 0$ for even $n$, there exists no $\kappa >2$ such that $f(n) \geq \kappa |\xi| n^{(k+l)/2}$, hence we cannot apply \cref{cor:faster}.
    For the same reason, $f(n) n^{-(k+l)/2}$ admits no asymptotic expansion $\kappa + o(1)$, and we cannot apply \cref{cor:a-deficiency}. Thus, we cannot draw any conclusion on the essential self-adjointness with the methods at hand. Still, by \cref{prop:conjugation} we do know that $\Aop$ has self-adjoint extensions.
\end{example}

\section{Concluding remarks}\label{sec:conclusion}

In this work, we have analyzed the essential self-adjointness of a broad class of operators of the form \( A_{k,l}(f) = \xi (a^\dagger)^k a^l + \xi^* (a^\dagger)^l a^k + f(a^\dagger a) \), defined on the linear span \( \mathcal{D}_0 \) of Fock states. These operators generalize higher-order squeezing Hamiltonians and include, as special cases, relevant models with Kerr nonlinearities or field energy terms.
Our main findings identify a sharp transition in the self-adjointness behavior governed by the asymptotic growth of \( f(n) \) compared to \( n^{(k+l)/2} \). When \( k + l < 3 \), essential self-adjointness always holds. For \( k + l \geq 3 \), the operator is essentially self-adjoint if \( f(n) \) grows sufficiently fast; otherwise, it admits a nontrivial space of self-adjoint extensions, whose structure we explicitly characterize by mapping the problem in the realm of recurrence relations and utilizing the Birkhoff–Trjitzinsky theory to characterize the square summability of the solutions. 
We also showed that $\Aop$ always has at least one self-adjoint extension for arbitrary choices of parameters.

In particular, our results can be applied to the case of higher-order squeezing operators possibly supplemented by field terms that are polynomial in the number operator, e.g. Kerr-type terms. We showed that such terms can act as a regularizing mechanism: when the field term dominates asymptotically (e.g., when its order exceeds that of the squeezing part), it restores both essential self-adjointness and boundedness from below, ensuring the existence of a well-defined quantum evolution.

Physically, these results clarify when higher-order squeezing Hamiltonians define unambiguous quantum dynamics and when additional physical input---to be mathematically encoded via a properly chosen unitary matrix appearing in the ``singular'' part $\mathcal{D}_U$ of the domain, cf. \cref{eq:singularpart}---is required to complete the description. Notably, the presence of field terms dominating in the large-excitation regime can provide a natural mechanism for selecting well-posed dynamics.

Several directions remain open for future investigation. Most notably, the strategy applied in this paper to parametrize essential self-adjoint extensions can be extended to more general classes of bosonic operators defined on \( \mathcal{D}_0 \). In more complex cases, the resulting recurrence relation will typically be of higher order, requiring the full machinery of Birkhoff–Trjitzinsky theory to analyze its asymptotics. This will likely require an explicit analysis of higher-order recurrence relations possibly involving coefficients that admit asymptotic expansions in powers of $r^{-1/\omega}$ with $\omega > 2$.
Additionally, a full characterization of the spectral properties of the operators considered here—including the nature of their spectra and possible bound states—would be of high value. These questions shall be the object of future research.

Finally, in addition to their conceptual implications, our results bear directly on the practice of numerical simulation. Since any simulation of bosonic models requires a finite-dimensional truncation—typically performed in the Fock basis $(\phi_n)_{n\in\mathbb{N}}$—the convergence of such simulations crucially depends on whether this basis forms a core for the operator in question. When it does, convergence to the correct unitary dynamics is guaranteed. In the non-essentially self-adjoint cases we study, where Fock states are not a core, the operators are not bounded from below. This rules out the applicability of recent convergence results to the dynamics generated by the Friedrichs extension \cite{fischer2025wrong}, leaving the outcome of numerical simulations uncertain.

This problem is studied in our companion paper \cite{ashhab2025truncation}, where it is shown that simulations of higher-order squeezing operators without additional terms exhibit an unexpected behavior: rather than converging, they oscillate between two distinct unitary groups, which we identify as arising from different self-adjoint extensions---precisely, those corresponding to the two choices $U=\pm I$ in the parametrization of \cref{thm:a-deficiency}. The addition of field terms, such as Kerr contributions, then plays a decisive role: when they dominate asymptotically, they restore essential self-adjointness and stabilize the numerical output. In this way, our analysis not only sheds light on the structure of higher-order bosonic Hamiltonians, but also provides practically relevant guidance for their computational study.

\subsection*{Acknowledgments} The authors thank Sahel Ashhab, Daniel Braak, and Rémi Robin for our scientific discussions and exchanges of ideas.
DL acknowledges financial support by Friedrich-Alexander-Universität Erlangen-Nürnberg through the funding program “Emerging Talent Initiative” (ETI), and was partially supported by the project TEC-2024/COM-84 QUITEMAD-CM.

\medskip

\appendix

\section{Proof of \texorpdfstring{\cref{prop:birkhoff}}{Proposition 4.5}}\label{sec:appendix}

Here we will provide the proof of \cref{prop:birkhoff}. As discussed in \cref{sec:recurrence}, this proposition is a specialization of the Birkhoff--Trjitzinsky theory~\cite{birkhoff-analytictheorysingular-1933}, which  provides precise conditions under which formal series solutions of recurrence relations correspond to asymptotic expansions of actual solutions. Since the original statement and proof of this result are known to be technically intricate and challenging to parse (see, e.g.,~\cite{wong-asymptoticexpansionssecondorder-1992,elaydi-introductiondifferenceequations-2005,braaksma-summationformalsolutions-2000,immink-reductioncanonicalforms-1991,wong-asymptoticexpansionssecondorder-1992}), we shall follow the more accessible presentation by Wimp and Zeilberger~\cite{wimp-resurrectingasymptoticslinear-1985}, directly applying it to second-order recurrence relations in the form~\eqref{eq:recurrence-birkhoff}.

\begin{theorem}[\cite{wimp-resurrectingasymptoticslinear-1985}]
\label{thm:birkhoff-general}
     Consider a recurrence relation in the form~\eqref{eq:recurrence-birkhoff}. We assume that $a(r)$ and $b(r)$ have asymptotic expansion of the form
    \begin{equation}
        a(r) \sim r^{J/\omega} \sum_{s=0}^{\infty} a_s r^{-s/\omega}\, ,  \quad b(r) \sim  r^{K/\omega}\sum_{s=0}^{\infty} b_s r^{-s/\omega}
    \end{equation}
    for some $K,J \in \znum$, $\omega \in \nnum$, and $\omega \geq 1$.
    Then the following properties hold:
    \begin{enumerate}
        \item \cref{eq:recurrence-birkhoff} admits exactly two linearly independent formal series solutions in the form $\e^{Q^\pm_p(r)} s^\pm_p(r)$, where     \begin{equation}\label{eq:birkhoff_q}
        Q^\pm_p(r)  = \mu_0^\pm r \ln r + \sum_{j = 1}^p \mu^\pm_j r^{(p+1-j)/p} \, ,
        \end{equation}
and $s^\pm_p(r)$ is a formal series
        \begin{equation}
        \label{eq:birkhoff_s}
        s_p^\pm(r) = r^{\alpha_\pm} \sum_{j = 0}^t \sum_{s = 0}^\infty C^\pm_{sj} (\ln r)^j r^{(q_j-s)/p} \, ,
    \end{equation}
    where $t,p,q_j,\mu^\pm_0 p \in \nnum$, $\mu^\pm_j,\alpha_\pm,C^\pm_{sj} \in \cnum$, and $p = v \omega$ for some $v \in \nnum$.
    \item The two formal series solutions are asymptotic expansions for $r\to\infty$ of two linearly independent solutions $(d_r^\pm)_{r \in \nnum}$ of \cref{eq:recurrence-birkhoff}:
    \begin{equation}
        d_r^\pm \sim \e^{Q^\pm_p(r)} s^\pm_p(r)\, .
    \end{equation}
    \end{enumerate}
\end{theorem}
In this specific setting, $\e^{Q^\pm_p(r)} s^\pm_p(r)$ is a formal series solution of \cref{eq:recurrence-birkhoff} if, upon substituting it into \cref{eq:recurrence-birkhoff} and factoring out the exponential, the coefficient of each term
\begin{equation}
    r^{\alpha_\pm + (q_j - s)/p} (\ln r)^j, \qquad 0 \le j \le t,
\end{equation}
vanishes identically.

\begin{remark}
  The special case $\omega = 1$ (asymptotic expansions in powers of $r^{-1}$) is treated explicitly in \cite{wong-asymptoticexpansionssecondorder-1992}. This would suffice for the purpose of this paper in the cases of even $k+l$, but would leave out the odd cases, for which $\omega=2$ is necessary.
The original work \cite{birkhoff-analytictheorysingular-1933} also covers higher-order recurrence relations and coefficients whose asymptotic expansions include a finite number of polynomially growing terms $r^s$.
For a detailed exposition with illustrative examples, see Ref.~\cite{wimp-resurrectingasymptoticslinear-1985}.
\end{remark}

We can now proceed with the proof of \cref{prop:birkhoff}.

\begin{proof}[Proof of \cref{prop:birkhoff}]
    As the assumptions of \cref{thm:birkhoff-general} are satisfied with $K = J = 0$ and $\omega = 2$, our strategy will be the following. We shall search for two formal series solutions of \cref{eq:recurrence-birkhoff} as in \cref{eq:birkhoff-solution}, and
    \begin{enumerate}
        \item We will show that \cref{eq:birkhoff-solution} does indeed reduce to the general form of \cref{thm:birkhoff-general} with $p = 2v$ for some $v \in \nnum$ and some value of the parameters;
        \item We will then show that \cref{eq:birkhoff-solution} is a formal series solution of \cref{eq:recurrence-birkhoff} if the parameters $\Omega_\pm$ and $\alpha_\pm$ are as given by \cref{eq:birkhoff-omega,eq:birkhoff-alpha} respectively.
    \end{enumerate}
    Since, by \cref{thm:birkhoff-general}, there exist exactly two linearly independent formal series solutions in this form, this will imply our claim.\smallskip

    For the first point, it suffices to notice that, by choosing $v = 1$ (thus $p = 2$), $t = 0$, $\mu_0^\pm = 0$, $q_0=0$, $\mu_1^\pm=\log\rho_\pm$ (we use the principal determination of the logarithm), and using the notation $\Omega_\pm = \mu_2^\pm$ and $C^\pm_s=C^\pm_{0s}$, \cref{eq:birkhoff_q,eq:birkhoff_s} reduce to
\begin{align}
        Q_2^\pm(r) &=  (\log \rho_\pm) r + \Omega_\pm r^{1/2}\, ,\\
        s_2^\pm(r) & = r^{\alpha_\pm} \sum_{s = 0}^\infty C^\pm_{s} r^{-s/2} \, ,
    \end{align}
    and therefore the formal series $\e^{Q_2^\pm(r)}s_2^\pm(r)$ indeed reduces to \cref{eq:birkhoff-solution}, and are linearly independent by $\rho_+\neq\rho_-$.
    
    We proceed with the second point. Plugging the formal series~\eqref{eq:birkhoff-solution} and the asymptotic expansions of $a(r)$ and $b(r)$ into \cref{eq:recurrence-birkhoff}, we get the \textit{formal} equality
    \begin{align}
    	\rho^{r+2}_\pm \e^{\Omega_\pm (r+2)^{1/2}}&  \sum_{s=0}^{\infty} C_s^\pm (r+2)^{\alpha_\pm-s/2} + \rho^{r+1}_\pm \e^{\Omega_\pm (r+1)^{1/2}} \sum_{s=0}^\infty \sum_{s' = 0}^\infty  C_s^\pm  (r+1)^{\alpha_\pm-s/2}a_{s'} r^{-s'/2} \\
    	& + \rho^{r}_\pm \e^{\Omega_\pm r^{1/2}} \sum_{s=0}^\infty \sum_{s' = 0}^\infty   C_s^\pm  r^{\alpha_\pm-s/2}b_{s'}r^{-s'/2} = 0 \, .
    \end{align}
    As $b_0\neq0$, the two roots $\rho_\pm$ are both nonzero and we can factor out a term $\rho_\pm^r\e^{\Omega_\pm r^{1/2}} r^{\alpha_\pm}$ above. We can also rewrite $(r+\nu)^{\alpha_\pm-s/2} = r^{\alpha_\pm-s/2}(1+\nu/r)^{\alpha_\pm-s/2}$ for $\nu=1,2$. 
By these manipulations, we get    
\begin{align}\label{eq:termbyterm}
    	\rho^{2}_\pm & \e^{\Omega_\pm (r+2)^{1/2}-\Omega_\pm r^{1/2}} \sum_{s=0}^{\infty} C_s^\pm r^{-s/2} \left(1+\frac{2}{r}\right)^{\alpha_\pm-s/2}\\
    	& + \rho_\pm \e^{\Omega_\pm (r+1)^{1/2}-\Omega_\pm r^{1/2}} \sum_{s=0}^\infty \sum_{s' = 0}^\infty  C_s^\pm a_{s'}  r^{-s/2-s'/2}\left(1+\frac{1}{r}\right)^{\alpha_\pm-s/2}\\& +  \sum_{s=0}^\infty \sum_{s' = 0}^\infty   C_s^\pm  b_{s'}r^{-s/2-s'/2} = 0 \, .
    \end{align}
    Following the approach in \cite{wong-asymptoticexpansionssecondorder-1992,wimp-resurrectingasymptoticslinear-1985}, we can manipulate this equality by some elementary series expansions. For $\nu=1,2$, we have
    \begin{equation}
    	\left(1+\frac{\nu}{r}\right)^{1/2} = 1 + \frac{\nu}{2r}-\frac{\nu^2}{8r^2} + o(r^{-2})\, ,
    \end{equation}
    and therefore
    \begin{equation}
    	(r+\nu)^{1/2} -r^{1/2}=  r^{1/2}\left( \frac{\nu}{2r}-\frac{\nu^2}{8r^2} + o(r^{-2})\right) = \frac{\nu}{2}r^{-1/2} -\frac{\nu^2}{8}r^{-3/2} + o(r^{-3/2}) \, ,
    \end{equation}
    whence implying 
    \begin{equation}
    	\label{proofeq:beta-expansion}
    	\e^{\Omega_\pm(r+\nu)^{1/2} -\Omega_\pm r^{1/2}} = 1+ \frac{\Omega_\pm \nu}{2}r^{-1/2} + \frac{\Omega_\pm^2 \nu^2}{8}r^{-1} + o(r^{-1}) \, .
    \end{equation}
    Analogously, by the binomial theorem, we obtain
    \begin{equation}
    	\label{proofeq:alpha-expansion}
    	\left(1+\frac{\nu}{r}\right)^{\alpha_\pm-s/2} = 1+ \nu \left(\alpha_\pm-\frac{s}{2}\right)r^{-1} + o(r^{-1}) \, .
    \end{equation}
    Inserting \cref{proofeq:beta-expansion,proofeq:alpha-expansion} into \cref{eq:termbyterm}, we obtain a formal series of powers of $r^{-1/2}$:
     \begin{align}
    \label{proofeq:termbyterm-expanded}
        & \rho^{2}_\pm  \left(1+\Omega_\pm r^{-1/2}+\frac{\Omega_\pm^2}{2}r^{-1} + o(r^{-1})\right) \sum_{s=0}^{\infty} C_s^\pm r^{-s/2} \left(1+2\left(\alpha_\pm-\frac{s}{2}\right)r^{-1}  + o(r^{-1}) \right)\\
    	& + \rho_\pm\left(1+\frac{\Omega_\pm}{2} r^{-1/2}+\frac{\Omega_\pm^2}{8}r^{-1} + o(r^{-1})\right) \sum_{s=0}^\infty \sum_{s' = 0}^\infty  C_s^\pm  a_{s'} r^{-s/2-s'/2}\left(1+\left(\alpha_\pm-\frac{s}{2}\right)r^{-1}  + o(r^{-1}) \right)\nonumber\\& +  \sum_{s=0}^\infty \sum_{s' = 0}^\infty  C_s^\pm  b_{s'} r^{-s/2-s'/2} = 0 \, . \nonumber
    \end{align}
    Factoring out the terms $C^\pm_s r^{-s/2}$ in each addend, we can reorder the expression above as such: 
    \begin{equation}
        \sum_{s = 0}^\infty C^\pm_s r^{-s/2}\left(f^\pm_0(s) + f^\pm_1(s) r^{-1/2} + f^\pm_2(s) r^{-1} + \ldots\right) = \sum_{s=0}^\infty C^\pm_sr^{-s/2}\sum_{s'=0}^\infty f_{s'}^\pm(s)r^{-s'/2}=0\,,
    \end{equation}
    where the coefficient functions $(f^\pm_j(s))_{j \in \nnum}$ depend on the parameters $\rho_\pm,\Omega_\pm,\alpha_\pm$ as well as $(a_{s'})_{s' \in \nnum}$ and $(b_{s'})_{s' \in \nnum}$. In particular, an explicit calculations yields the following result for the first three orders:
    \begin{align}
        f^\pm_0(s) &= \rho_\pm^2 + \rho_\pm a_0 + b_0;\\
        f^\pm_1(s)&=\rho_\pm^2 \Omega_\pm + \rho_\pm \frac{\Omega_\pm}{2} a_0 + \rho_\pm a_1 +b_1\, , \\
        f^\pm_2(s)&=\rho_\pm^2 \frac{\Omega_\pm^2}{2} + 2\rho_\pm^2\left(\alpha_\pm - \frac{s}{2}\right) + \rho_\pm \frac{\Omega_\pm^2}{8}a_0 + \rho_\pm \frac{\Omega_\pm}{2}a_1 +\rho_\pm a_2 + \rho_\pm a_0 \left(\alpha_\pm -\frac{s}{2}\right) + b_2\,.
    \end{align}
    Both $f^\pm_0(s)\equiv f_0^\pm$ and $f^\pm_1(s)\equiv f^\pm_1$ are identically zero. Indeed, $f_0^\pm=0$ by virtue of the characteristic equation~\eqref{eq:characteristic}, whose roots are precisely $\rho_\pm$, and $f^\pm_1$ also vanishes if and only if
    \begin{equation}
        \label{proofeq:omega-expression}
        \Omega_\pm = -\frac{a_1 \rho_\pm+b_1}{a_0 \rho_\pm/2+\rho_\pm^2} \, ,
    \end{equation}
    which is \cref{eq:birkhoff-omega}. Note that
    the denominator is nonzero under our assumptions since, by $b_0 \neq 0$, we get $\rho_\pm \neq 0$, and from $\rho_+ \neq \rho_-$ follows $a_0/2+\rho_\pm \neq 0$. As for $f^\pm_2(s)$, we can show that, if $\alpha_\pm$ is given by \cref{proofeq:alpha-expansion}, then $f^\pm_2(0)=0$. Indeed, by direct algebraic manipulation, \cref{eq:birkhoff-alpha} is equivalent to 
    \begin{align}\label{proofeq:alpha-expression} 
    \alpha_\pm & = \frac{a_2 \rho_\pm+b_2}{a_0 \rho_\pm + 2 b_0} - \frac{\Omega_\pm^2(\rho_\pm/2+ a_0/8)+ \Omega_\pm a_1/2}{2\rho_\pm+a_0 }\\
            & = -\frac{a_2 \rho_\pm+b_2}{2 \rho_\pm^2+ a_0 \rho_\pm} - \frac{\rho_\pm^2 \frac{\Omega_\pm^2}{2}+ \rho_\pm \frac{\Omega^2}{8}a_0 + \rho_\pm \frac{\Omega_\pm}{2}a_1 }{2 \rho_\pm^2+ a_0 \rho_\pm} \, ,
    \end{align}
    hence $f^\pm_2(s)$ vanishes precisely at $s=0$. Again, the expression for $\alpha_\pm$ is well-defined as $b_0 \neq 0$ and $\rho_+ \neq \rho_-$.

    We thus proved that, with our choice of parameters, $f^\pm_0=f^\pm_1=0$, and $f^\pm_2(0)=0$; since $s\mapsto f^\pm_2(s)$ is a linear function, $s=0$ is its only root, so that $f^\pm_2(s)\neq0$ for $s\neq0$.
    
    We now rearrange the double formal sum above by grouping together all terms with the same value of $s+s'$; in this form, we obtain a double sum in two new indices $\sigma=0,1,2,\ldots$ and $j=0,1,\ldots,\sigma$,
    \begin{equation}
        \sum_{\sigma=0}^\infty r^{-\sigma/2}\sum_{j=0}^\sigma f^\pm_{\sigma-j}(j)C^\pm_j=0.
    \end{equation}
Therefore, \cref{eq:birkhoff-solution} defines a formal series solution of the recurrence relation if and only if
\begin{equation}\label{proofeq:order-by-order}
    \sum_{j=0}^\sigma f^\pm_{\sigma-j}(j)C^\pm_j=0\qquad\text{for all }\sigma\in\mathbb{N}.
\end{equation}
We now show that, with our choice of parameters $\rho_\pm,\Omega_\pm$ and $\alpha_\pm$, one can always find $(C_s^\pm)_{s \in \nnum}$ such that \cref{proofeq:order-by-order} is formally satisfied at all orders $r^0,r^{-1/2},r^{-1}, \dots$. To this end, note that \cref{proofeq:order-by-order} is satisfied for $\sigma=0,1,2$ independently of the values of $(C_s^\pm)_{s\in\nnum}$ since, for those values of $\sigma$, it reduces respectively to
\begin{align}
    f_0^\pm(0)C_0^\pm&=0;\\
    f_1^\pm(0)C_0^\pm+f_0^\pm(1)C_1^\pm&=0;\\
    f_2^\pm(0)C_0^\pm+f_1^\pm(1)C_1^\pm+f_0^\pm(2)C_2^\pm&=0,
\end{align}
but, as proven above, $f^\pm_0$ and $f^\pm_1$ are identically zero, and $f^\pm_2(0)=0$. For $\sigma\geq3$, \cref{proofeq:order-by-order} reads
\begin{equation}
    f_\sigma^\pm(0)C_0^\pm+f_{\sigma-1}^\pm(1)C_1^\pm+\ldots+f^\pm_3(\sigma-3)C_{\sigma-3}^\pm+f_2^\pm(\sigma-2)C_{\sigma-2}^\pm=0,
\end{equation}
and, since we also proved $f^\pm_2(s)\neq0$ for $s\neq0$, and $\sigma-2\neq0$, we can solve this equation in $C_{\sigma-2}^\pm$:
\begin{equation}\label{eq:happyending}
    C_{\sigma-2}^\pm=-\frac{1}{f_2^\pm(\sigma-2)}\sum_{j=0}^{\sigma-3}f^\pm_{\sigma-j}(j)C_j^\pm.
\end{equation}
This equation fixes all coefficients $C_1^\pm,C_2^\pm,\ldots$ in terms of $C_0^\pm$: indeed, for $\sigma=3$, \cref{eq:happyending} fixes $C_1^\pm$ in terms of $C_0^\pm$, for $\sigma=4$ it fixes $C_2^\pm$ in terms of $C_1^\pm$, and so on. That is, our ansatz~\eqref{eq:birkhoff-solution} yields a valid formal series solution of \cref{eq:recurrence-birkhoff} uniquely defined up to a multiplicative constant.

This completes the construction of two linearly independent formal series solutions of \cref{eq:recurrence-birkhoff} in the form stated in \cref{thm:birkhoff-general}, from which the claim follows by the argument given at the beginning of the proof.
\end{proof}

\AtNextBibliography{\small}	\DeclareFieldFormat{pages}{#1}\sloppy 
	\printbibliography
    
\end{document}